\documentclass[sigconf,nonacm]{acmart}

\usepackage[utf8]{inputenc}
\usepackage{amsthm}

\usepackage{xspace}

\usepackage[ruled, linesnumbered, vlined]{algorithm2e}
\SetKwInput{KwInput}{Input}
\SetKwInput{KwOutput}{Output}
\SetKwComment{tcc}{\# }{}
\SetKwComment{tcp}{\# }{}
\SetKwIF{If}{ElseIf}{Else}{if}{}{else if}{else}{end if}
\SetCommentSty{texttt}

\DeclareMathOperator{\cond}{cond}
\renewcommand{\dot}[1]{\overset{\,_{\mbox{\Large .}}}{#1}}

\makeatletter
\tagsleft@true
\makeatother

\title{Fast evaluation and root finding for polynomials with floating-point coefficients}
\author{Rémi Imbach}
\affiliation{
  \institution{Université de Lorraine, CNRS, Inria, LORIA}
  \city{F-54000 Nancy}
  \country{France}
}
\author{Guillaume Moroz}
\email{guillaume.moroz@inria.fr}
\affiliation{
  \institution{Université de Lorraine, CNRS, Inria, LORIA}
  \city{F-54000 Nancy}
  \country{France}
}

\begin{abstract}

  Evaluating or finding the roots of a polynomial $f(z) = f_0 + \cdots +
  f_d z^d$ with floating-point number coefficients is a ubiquitous
  problem. By using a piecewise approximation of $f$ obtained with a
  careful use of the Newton polygon of $f$, we improve state-of-the-art
  upper bounds on the number of operations to evaluate and find the
  roots of a polynomial. In particular, if the coefficients of $f$ are
  given with $m$ significant bits, we provide for the first time an
  algorithm that finds all the roots of $f$ with a relative condition
  number lower than $2^m$, using a number of bit operations quasi-linear
  in the bit-size of the floating-point representation of $f$. Notably,
  our new approach handles efficiently polynomials with coefficients
  ranging from $2^{-d}$ to $2^d$, both in theory and in practice.

\end{abstract}

\acmConference{ISSAC 2023}{2023}{Norway}

\begin{document}
\maketitle
 
Evaluating or finding the roots of a polynomial 
with
coefficients represented as floating-point numbers
is widely used.
Given two positive integer
bounds $m$ and $\tau$, a floating number can be represented as $a2^b$ where $a$ and
$b$ are integers with $|a| \leq 2^m$ and $|b|\leq \tau$. This representation notably allows
to use a constant size $m+\log_2 \tau$ to represent numbers with
different orders of magnitude \cite{higham2002accuracy,muller2018handbook,brent2010modern}.

In the literature, analysing the complexity for approximating the roots
of a polynomial is usually done by considering a fixed-point
representation for the
coefficients (\cite{pan2002univariate,MSWjsc15,becker2018near,moroz2022new} and references
therein). In those analyses, the coefficients are
either integers or represented as fixed-point numbers with a uniform
error on all the coefficients. The drawback of those approaches is that
finding the root of the polynomial $2^\tau -2^{-\tau} z = 0$ will
require $\widetilde O(\tau)$ bits operations where $\widetilde O(\cdot)$
means that we omit the logarithmic factors. On the other hand, this
simple equation can be solved in $\widetilde O(\log \tau)$ bit
operations if we use floating-point
arithmetic.

Polynomial where the coefficients have different orders of magnitude appear in several applications.
For example, truncating a series expansion of a generalized
hypergeometric function such as $e^z$ to order $d$
may return a polynomial where the $k$-th coefficient has a
magnitude in $\varTheta(1/k!)$. The characteristic polynomial of a matrix
also has coefficients with different order of magnitude. Solving
multivariate polynomial systems can be done by eliminating variables and
reducing the problem to a univariate polynomial having coefficients with
different orders of magnitude \cite{BPRbook06}.

Some work in the literature address the problem of handling large orders of magnitudes for
the root-finding problem,
notably when using methods based on the Newton diagram and Graeffe
iterations (\cite{ostrowski1940recherches,schonhage1982fundamental,gourdon:inria-00074820,pan2000approximating,malajovich2001geometry}
among others). In
these approaches, intermediate values can have large order of magnitudes
and a process of normalisation can be applied to reduce their sizes
\cite{malajovich2001tangent,grau63reduction,henrici1993applied}. However, the algorithm they describe
is quadratic in the degree of the input polynomial.

\begin{figure}
  \includegraphics[width=0.9\linewidth]{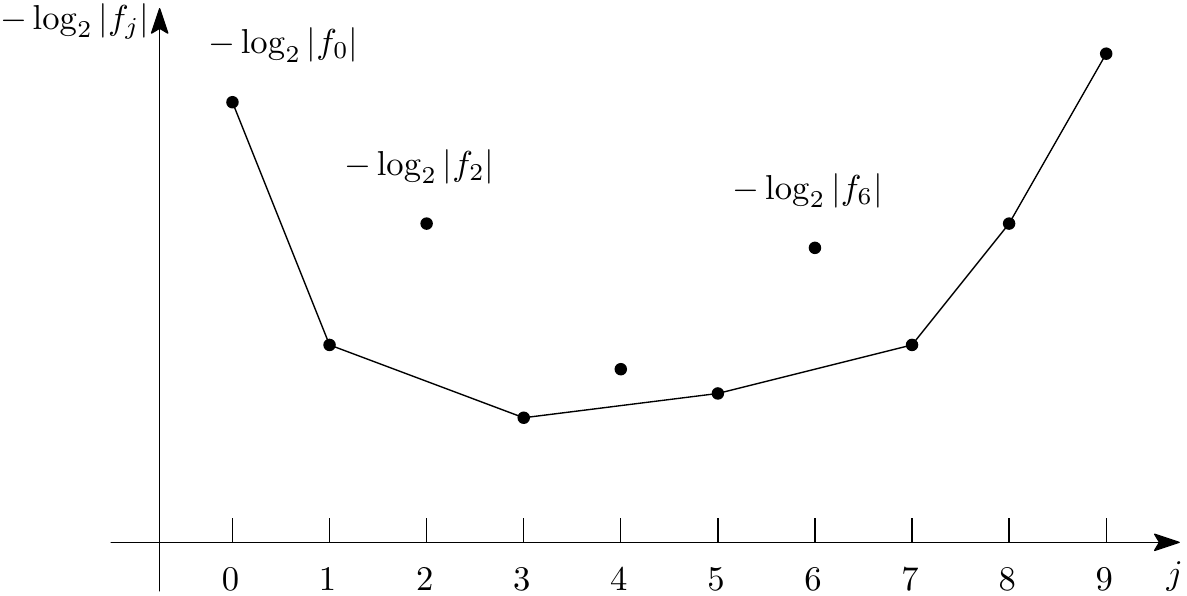}
  \caption{The Newton polygon of the polynomial $f$, using the valuation
           $-\log_2 |f_j|$ for the $j$-th coefficient.}
  \label{figure:newton_polygon}
  \vspace*{-0.5cm}
\end{figure}

Using the Newton polygon of $f$ we develop a new method adapted to the working
relative precision of the input coefficients. Let $f$ be a
polynomial of degree $d$ with floating-point coefficients of magnitude less than
$2^\tau$ and represented with $m$ significant bits. We will compute a
piecewise approximation of $f$ by $O(d/m)$ polynomials $g$ of degree
$O(m)$, where the coefficients of each polynomial $g$ is represented
with $O(m)$ significant digits. This representation will allow us to
improve state-of-the-art bounds on the problems of evaluating $f$ on
$d$ points and 
of finding its roots.

After stating formally our main results in
Section~\ref{sec:preliminaries}, we will describe in Sections~\ref{section_rings}
and~\ref{section_approximation} how to compute such a piecewise
approximation
in $\widetilde O(d(m+\log \tau))$ bit operations. Then we
will show how to use this data-structure to evaluate $f$
(Section~\ref{section_evaluation}) and to approximate or isolate the roots of $f$
(Section~\ref{section_isolation}). Finally we describe experimental
results of our prototype implementations on different families of
polynomials (Section~\ref{section_benchmarks}).

\section{Preliminaries}
\label{sec:preliminaries}
In the article, our input is the polynomial $f(z) = f_0 + \cdots + f_d
z^d$. We assume that the coefficients of $f$ are represented with
floating-point numbers and our goal is to design fast algorithms for the
evaluation or the root-finding problem, with the same error bounds as if
we had used classical algorithms with floating-point arithmetic.

For that we will first compute in Sections~\ref{section_rings}
and~\ref{section_approximation} a piecewise polynomial approximation of $f$
defined on a partition of the complex plane.
Before stating our main result, we recall two mathematical tools that
will be fundamental in our algorithms: the \emph{condition number} and
the \emph{Newton polygon}. 

\paragraph{Notations}
We will use the following notations to describe our piecewise polynomial
approximation.  For a complex number $\gamma$ and a positive real
$\rho$, we denote by $D(\gamma, \rho)$ the disk of center $\gamma$ and
radius $\rho$. For two positive real number $r_1,r_2$, we denote by
$R(r_1,r_2)$ the ring $D(0,r_2)\setminus D(0,r1)$. Moreover, for a
polynomial $f$ we let:
\begin{align*}
  \widetilde f(|z|) &= |f_0| + \cdots + |f_d| |z|^d, & 
  \|f\|_1 &= |f_0| + \cdots + |f_d|\\
  \hat f(|z|) &= \max_{0\leq j \leq d}(|f_j||z|^j), &
  f_{\ell,u}(z) &= f_l + \cdots + f_u z^{u-\ell}.
\end{align*}

\subsection{Relative condition number}
In floating-point representation, the errors are relative.
Let $f_\varepsilon$ be the polynomial obtained by adding a relative
error bounded by $\varepsilon$ on the coefficients of $f$. This amounts
to multiply all the coefficients $f_j$ by $(1+\varepsilon_j)$ with
$|\varepsilon_j| \leq \varepsilon$.

In the \emph{evaluation problem}, for any complex point
$z$ we have
$$|f(z) - f_\varepsilon(z)| \leq \varepsilon \widetilde f(|z|)$$
and this bound is tight. In Theorem~\ref{thm:evaluation} we focus on the problem
of finding the approximate 
value
of $f(z)$ with the same error bound.

For the \emph{root-finding problem}, we can also define the relative condition
number \cite{graillat2008accurate},\cite[\S 1.6]{higham2002accuracy}
associated to a relative perturbation of the coefficients of $f$. More
precisely, if $\zeta$ is a root of $f$, and $\dot \zeta$ is the closest
root of $f_\varepsilon$, the relative condition number of $\zeta$
denoted by $\cond(f,\zeta)$ is
defined as $\lim_{\varepsilon \rightarrow 0} |\zeta - \dot
\zeta|/(|\zeta|\varepsilon)$. Using the Taylor expansion $f_{\varepsilon}(\zeta) = (\zeta-\dot
\zeta)f'_\varepsilon(\zeta) + O(|\zeta-\dot\zeta|^2)$ and
$| f(\zeta) - f_\varepsilon(\zeta)|\leq \varepsilon \widetilde f(|\zeta|)$, this leads to:
$$\cond(f,\zeta) = \frac{\widetilde f(|\zeta|)}{|\zeta||f'(\zeta)|}.$$
This means that the number of bits of precision required to compute
the first significant bit of $\zeta$ is in $O(\log(\cond(f,\zeta))$. In
Section~\ref{section_isolation}, we focus on the problem of finding the
roots of $f$ with a number of bit operations quasi-linear in
$\log(\cond(f,\zeta))$ and quasi-linear in $d$.

\subsection{Newton polygon}
Given a non-zero complex coefficient $f_j$ of the polynomial $f$, we can
associate it with the value  $-\log_2 |f_j|$. Then for each index
$0\leq j \leq d$, let $P_j$ be the point $(j, -\log_2 |f_j|)$. The
\emph{Newton polygon} of $f$ is the lower convex hull of the set of
points $P_j$, as illustrated on Figure~\ref{figure:newton_polygon}.
It has been used to give a rough first estimation of the
modules of the roots of a polynomial
\cite{gourdon:inria-00074820, ostrowski1940recherches, schonhage1982fundamental, pan2000approximating, malajovich2001geometry, BAna00, bini2014solving, imbach2021root}.

We will denote by $H$ the Newton polygon of $f$, and use it 
to partition the complex plane in rings where the variation of
$\widetilde f$ is bounded (Section~\ref{section_rings}).
%


%
%
%
%
%
%
%

\subsection{Main results}

We can now state our main results, with some reasonable assumptions on
the precision such as $m>\log d$ for a simpler presentation.

The first theorem shows that we can evaluate a polynomial of degree $d$
on $d$ points with a complexity quasi-linear in $d$ times the size of a
floating-point representation of the coefficients. In particular, if the
magnitudes of the coefficients are bounded by $2^\tau$, this means that
our algorithm will be quasi-linear in $d\log\tau$. This is an improvement over a recent
result \cite{moroz2022new} where the complexity was quasi-linear in $d\tau$.

\begin{theorem}[Evaluation]
  \label{thm:evaluation}
  Let $f$ be a polynomial of degree $d$ with complex floating-point coefficients
  of magnitudes less than $2^{\tau}$ and let $m$ be a positive integer.
  It is possible to compute in $\widetilde O(d(m+\log \tau))$ bit
  operations a piecewise polynomial approximation $f_{pw}$
  such that
  $$|f(z) - f_{pw}(z)| \leq 2^{-m} \widetilde f(z)$$ 
  for all $z \in \mathbb
  C$. Moreover, assuming that $m>\log d$, and given $d$ complex points
  $z_1,\ldots,z_d$ 
  with
  $|z_j|\leq 2^{\tau}$, it is possible to evaluate:
  \begin{itemize}
    \item[(i)] $f_{pw}(z_1)$ in $\widetilde O(m(m+\log\tau))$ bit operations,
    \item[(ii)] $f_{pw}(z_1),\ldots,f_{pw}(z_d)$ in $\widetilde
      O(d(m+\log \tau))$ bit operations.
  \end{itemize}
\end{theorem}

The second theorem allows us to
find the roots of polynomials with a number of bit operations
quasi-linear in $\log \tau$,
and to handle polynomials with
large coefficients such as resultant
polynomials of two
bivariate polynomials (see
Section~\ref{subsec_numerical_results}).

\begin{theorem}[Root finding]
  \label{thm:rootfinding}
  Let $f$ be a polynomial of degree $d$ with complex floating-point coefficients
  of magnitudes less than $2^{\tau}$ and let $m$ be a positive integer.
  It is possible to compute the $m$ most significant bits of
  the roots $\zeta$ of $f$ such that $\cond(f,\zeta) \leq 2^m$ in
  $\widetilde O(d(m+\log \tau))$ bit operations.

  Moreover, if $\kappa = \max_{\zeta \mid f(\zeta)=0} \cond(f,\zeta)$,
  it is possible to compute the $m$ most significant bits of all the
  roots of $f$ in $\widetilde O(d(m+\log \tau + \log \kappa)$ bit
  operations.
\end{theorem}

Our approach consists in computing a piecewise approximation in two
steps. A first approximation is computed on a partition of the complex
plane in concentric rings, and is described in Section~\ref{section_rings}. Then we refine our piecewise
approximation by partitioning each ring uniformly in angular sectors, as
described in Section~\ref{section_approximation}.



%
%
\section{Piecewise approximation over rings}
\label{section_rings}

\begin{figure}
  \parbox[b]{0.75\linewidth}{
  \includegraphics[width=0.285\linewidth]{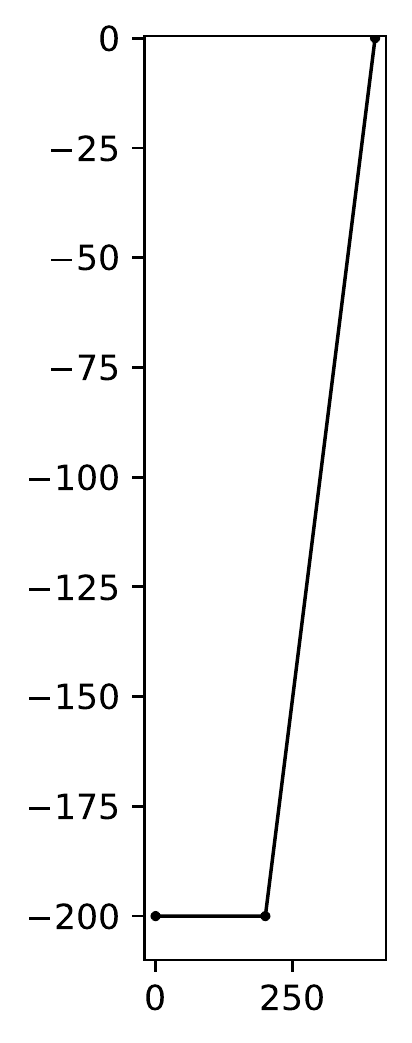}
  \includegraphics[width=0.715\linewidth]{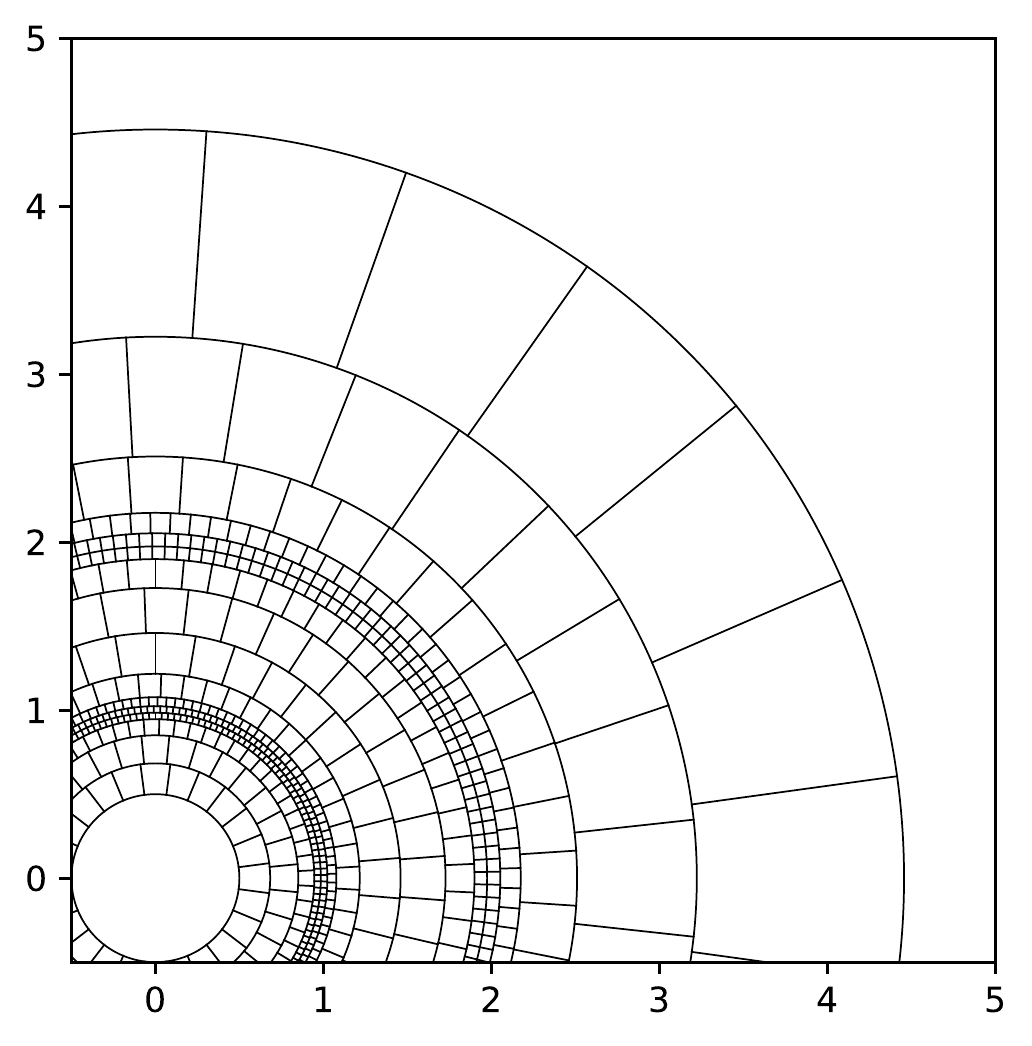}
  }
\caption{Newton polygon (left) 
and sector partition (right) associated to
the polynomial $(z^{200} - 1)(z^{200}-2^{200})$} 
\label{figure:example_polygon}
\vspace*{-0.5cm}
\end{figure}

Our first piecewise approximation is based on a subdivision of the complex
plane with concentric rings $R_n$ centered at the origin, as illustrated in
Figure~\ref{figure:example_polygon}. In
Algorithm~\ref{alg:rings}, we compute the rings and we select dominant
monomials of $f$. This returns a list of rings $R_n = R(r_n,r_{n+1})$
and a pair of indices $(\ell_n,u_n)$ associated to each ring such that
evaluating $z^{\ell_n}f_{\ell_n,u_n}(z)$ on $R_n$ is sufficient to
evaluate $f$ up to a given error bound. One of the goal of this
algorithm is to choose the rings such that the sum of the number of
monomials of the polynomials $f_{\ell_n,u_n}$ is linear in
$d$. The main result of this section is summarized in the following
Lemma.
 
\begin{lemma}
  \label{lem:rings}
  Given a polynomial $f$ of degree $d$ with floating-point coefficients
  with $m$-bits
  mantissa and magnitude less than $2^{\tau}$, it is
  possible to compute in $\widetilde O(d(m+\log \tau))$ bit operations
  $N$
  rings 
  $R_n = R(r_n, r_{n+1})$
  and
  extract $N$ polynomials $f_{\ell_n,u_n}$ of degrees $\delta_n$ such
  that:
  \begin{gather}
    \tag{i} 2^{\frac m 2}r_n^{\delta_n} \leq r_{n+1}^{\delta_n} \leq 2^m
    r_n^{\delta_n} \makebox[0pt]{\hspace{2cm} if $\delta_n \geq 1$, and}\\
    \tag{ii} \sum_{n=0}^{N-1} (\delta_n + 1) \leq 65 d + 1.
  \end{gather}
  And for all $z\in R_n$:
  \begin{gather}
    \tag{iii} |f(z) - z^{\ell_n}f_{\ell_n,u_n}(z)| \leq (d-\delta_n) 2^{-m} \hat f(|z|).
  \end{gather}
\end{lemma}



We use the Newton polygon associated to $f$
to compute such approximation polynomials,
which allows us to reduce computing rings and indices that satisfy Lemma~\ref{lem:rings}
to geometrical arguments.

In Section~\ref{sec:newtonalgo} we start by detailing
Algorithm~\ref{alg:rings} and prove the inequalities (i) and (iii) of
Lemma~\ref{lem:rings}.  Then in
Section~\ref{sec:newtoncomplexity} we prove the inequality (ii) that is
fundamental to bound the complexity of all the subsequent algorithms.

%
%
%

\subsection{Algorithm}
\label{sec:newtonalgo}
Our algorithm takes as input the Newton polygon $H$ of the polynomial $f$.
The Newton polygon is the lower convex hull
of the points $(j, -\log_2 |f_j|)$ for $0 \leq j \leq d$ and in the general
case, it can be computed in $O(d \log d)$ arithmetic operations
\cite{graham1983finding},
or $O(d)$ arithmetic operations when the points
are sorted, which is the case here.

\begin{figure}
  \includegraphics[width=0.9\linewidth]{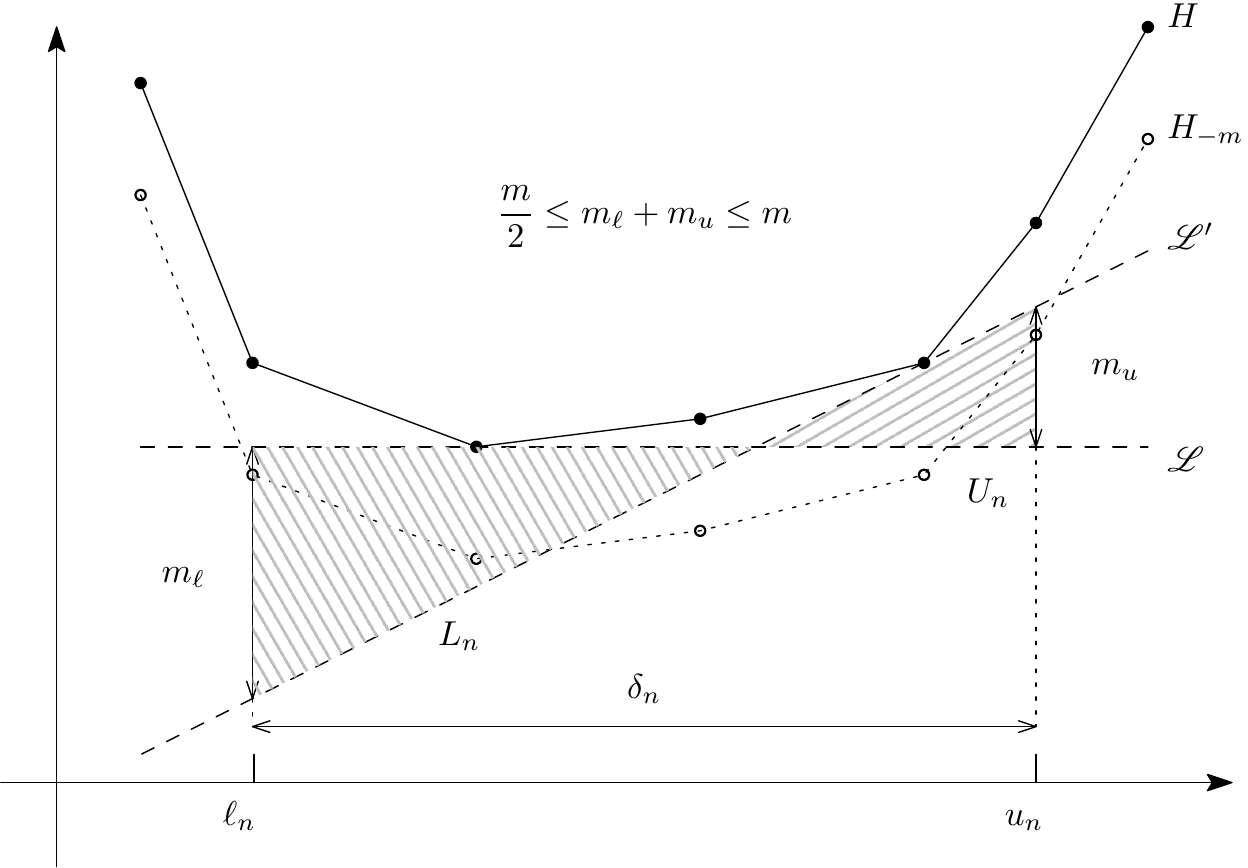}
  \caption{One step to sweep the line in Algorithm~\ref{alg:rings}.}
  \label{fig:onestep}
  \vspace*{-0.5cm}
\end{figure}

The main idea of Algorithm~\ref{alg:rings} is to swipe a line tangent to
$H$, as illustrated in Figure~\ref{fig:onestep}. Given a line $\mathscr L$ of slope $s$ tangent to $H$, the next
line $\mathscr L'$ of slope $s'$ tangent to $H$, with $s'>s$, is computed
informally as follow:
\begin{itemize}
  \item[1.] Let $H_{-m}$ be the polygon $H$ shifted vertically by $-m$.
  \item[2.] Let $\ell$ be the abscissa of the leftmost point of $H_{-m}$
    below $\mathscr L$.
  \item[3.] For the new line $\mathscr L'$, we denote by
    \begin{itemize}
      \item \hspace*{-0.1cm}$u$ the abscissa of the rightmost point of $H_{-m}$
        below $\mathscr L'$,
      \item \hspace*{-0.1cm}$m_\ell$   the vertical distance between $\mathscr L$ and
        $\mathscr L'$ at
        $\ell$,
      \item \hspace*{-0.1cm}$m_u$   the vertical distance between $\mathscr L$ and
        $\mathscr L'$ at $u$.
    \end{itemize}
  \item[4.] Starting from the slope of $\mathscr L$, increase the slope
    of $\mathscr L'$
    until
          $$m/2 \leq m_\ell + m_u \leq m.$$
\end{itemize}

In Algorithm~\ref{alg:rings}, line~\ref{line:l} 
addresses step 1 above. In
line~\ref{line:u} we compute the maximal abscissa
$u$ such that $m_\ell+m_u \leq m$. 
Then we choose $s'=s+m/(u-l+1)$ for the new slope of $\mathscr L'$;
this implies that $m_l+m_u = m(u-l)/(u-l+1) \leq m$, and as soon as
$u>l$, we have $m/2 \leq m_l+m_u \leq m$. Thus, if $\delta_n\geq 1$, we
have $(r_{n+1}/r_n)^{\delta_n} = 2^{m(u_n-\ell_n)/(u_n-\ell_n+1)}$ is
between $2^{m/2}$ and $2^m$,  which proves the inequality (i)
of Lemma~\ref{lem:rings}.

For the inequality (iii), let $z$ be such that $r_n\leq|z|\leq r_{n+1}$.
This means that $\log |z|$ is between $s_n$ and $s_{n+1}$. Let $\mathscr
L''$ be the line of slope $\log |z|$ tangent to $H$. It is below
$\mathscr L$ for the abscissae less than $\ell_n$, and below $\mathscr L'$
for the abscissae greater than $u_n$. In particular, let $P=(k,y_k)$ be
a vertex of $H$ contained in $\mathscr L''$, and let $Q = (j,y_j)$ be a point of
$H_{-m}$ such that $j<\ell_n$ or $j>u_n$ and $|f_j| \neq 0$. By construction, we have
$y_k = -\log_2 |f_k|$ and $y_j\leq-\log_2 |f_j|-m$. Moreover we know that
the line $\mathscr L''$ is below $P_j$. Since $\mathscr L''$ is a line
of slope $\log |z|$ passing through $P_k$, this implies that $y_k +
(j-k)\log_2 |z| \leq y_j$. Equivalently, this means that $-\log_2 |f_k|
+ (j-k)\log_2 |z| \leq -\log_2 |f_j|-m$. Taking the power of two, this
implies that $|z|^{j-k}/|f_k| \leq 2^{-m}/|f_j|$, which can be rewritten
as $|f_j||z|^j \leq 2^{-m}|f_k||z|^k \leq 2^{-m}\hat f(|z|)$. This
implies that for all $z\in R_n$ and all $j< \ell_n$ and all $j> u_n$, we
have $|f_j||z|^j \leq 2^{-m}\hat f(|z|)$, which in turn implies the
inequality (iii).





\begin{algorithm}
  \DontPrintSemicolon
  \KwInput{\hspace*{-3ex}%
    \begin{minipage}[t]{7.5cm}%
      \vspace{-0.5em}%
      \begin{itemize}
        \item[$m$:] a positive integer
        \item[$H$:] list of points $(j, h_j)$ on the Newton polygon
          of $f$\\
              for $0 \leq j \leq d$ 
      \end{itemize}         
    \end{minipage}
  }
  \KwOutput{Satisfying inequalities (i), (ii), (iii) of Lemma~\ref{lem:rings}
    \begin{minipage}[t]{7cm}%
      \begin{itemize}
        \item[$(r_n)_{0\leq n < N}$:] list of circle radii surrounding the rings
        \item[$(\ell_n)_{0\leq n < N}$:] lower indices of the dominant monomials
        \item[$(u_n)_{0\leq n < N}$:] upper indices of the dominant monomials
      \end{itemize}         
    \end{minipage}
  }
  \vspace{1em}
  \For{$j$ from $1$ to $d-1$}{
    $\mathcal{L} \gets$ lines tangent to $H$ passing through $(j, h_j-m)$\;
    $S^{0}_j \gets $lower slope of the $2$ lines in $\mathcal{L}$\;
    $S^{1}_j \gets $upper slope of the $2$ lines in $\mathcal{L}$\;
  }
  start $\gets$ slope of line through $(0,h_0)$ and $(1,h_1-m)$\;
  end $\gets$ slope of line through $(d-1,h_{d-1})$ and $(d,h_d-m)$\;
  $n, s_n \gets 1$, start\;
  $r_0$, $r_1$, $\ell_0$, $u_0 \gets 0, 2^{\text{start}}, 0, 0$ \;
  \While{$s_n < $ end}{
    $\ell_n \gets $ minimal index $\ell$ such that $S^{1}_\ell >
    s_n$\label{line:l}\;
    $u_n \gets $ maximal index $u$ such that $(u-\ell_n)(S^{0}_{u} - s) < m$\label{line:u}\;
    $s_{n+1} \gets s_{n} + \frac{m}{u_n-\ell_n+1}\label{line:s}$
    \lIf*{$u_n>\ell_n$}{\lElse{$S^{0}_{u_n+1}$}}
    $r_{n+1} \gets 2^{s_{n+1}}$\;
    $n \gets n+1$ \;
  }
  $r_{n+1}, \ell_n, u_n \gets +\infty, d, d$\;
  $N \gets n+1$\;
  \Return $(r_n)_{0\leq n \leq N}$, $(\ell_n)_{0\leq n < N}$, $(u_n)_{0\leq n < N}$\;

  \caption{Rings and dominant monomials}
  \label{alg:rings}
\end{algorithm}


\subsection{Complexity}
 \label{sec:newtoncomplexity}

%
%
%
%

We can now prove the inequality (ii) that will allow us notably to bound the
number of bit operations of Algorithm~\ref{alg:rings} in
Section~\ref{sec:ringcomplexity}.

\subsubsection{Bound on cumulated number of monomials}

\begin{figure}
  \includegraphics[width=0.9\linewidth]{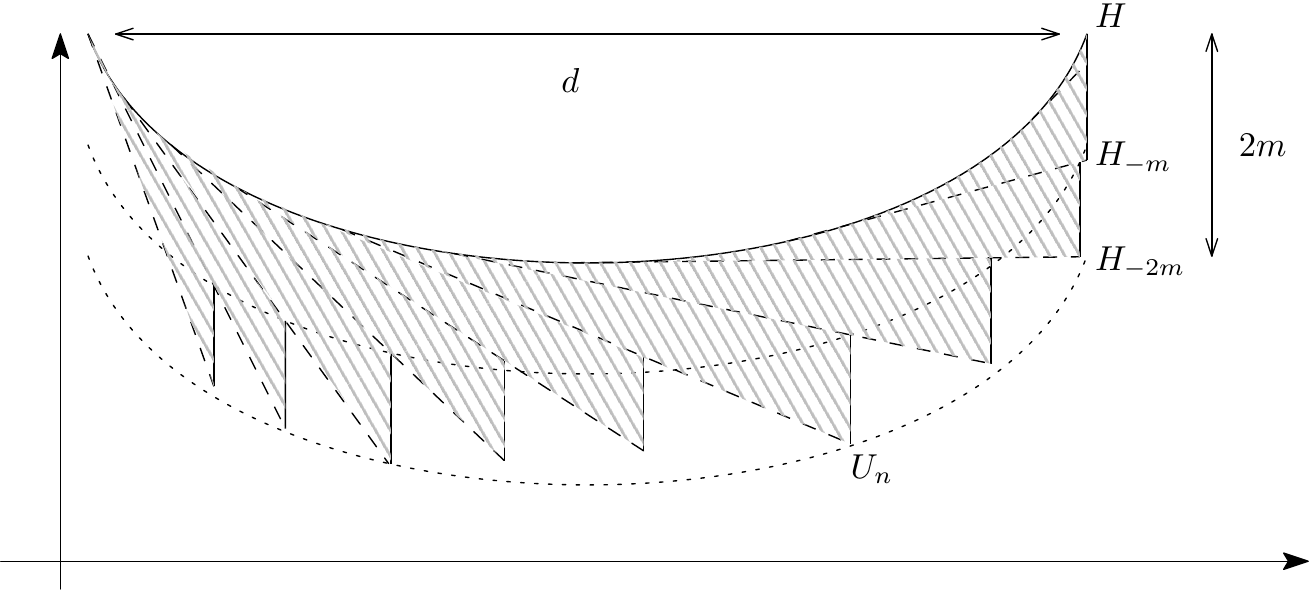}
  \caption{$N$ steps to sweep the line in Algorithm~\ref{alg:rings}.}
  \label{fig:Nsteps}
  \vspace*{-0.5cm}
\end{figure}

By the way the $\delta_n$ are computed, it is not obvious that their sum
is linear in $d$. The proof comes from a geometrical argument. For each
step, we can define an area that has a size in $\Omega(\delta_n m)$, and
we will show that the sum of the area is in $O(dm)$, as illustrated on
Figure~\ref{fig:Nsteps}. More precisely, for $0 \leq n < N$ such that
$\delta_n \geq 1$, let $V_n$ be the vertical band between the abscissae
$\ell_n$ and $u_n$. Letting $\mathscr L$ and $\mathscr L'$ be the lines of slopes $s_n$ and
$s_{n+1}$ respectively, we define two regions:
\begin{itemize}
  \item $U_n$ the set of points of $V_n$ above $\mathscr L$ and below $\mathscr L'$
  \item $L_n$ the set of points of $V_n$ below $\mathscr L$ and above $\mathscr L'$
\end{itemize}

First, since $V_n$ has width $\delta_n$, and the difference of the slopes between $\mathscr L$
and $\mathscr L'$ is $m/(\delta_n+1) \geq m/(2\delta_n)$, the sum of the areas of $L_n$
and $U_n$ is at least $\delta_n m / 8$. 

Then, let $U$ be the union of the $U_n$ for $0\leq n<N$ and
$\delta_n\geq1$. By construction, for different indices $n$, the $U_n$
are pairwise distinct, such that the area of $U$ is greater than
the sum of the area of $U_n$. Similarly, letting $L$ be the union of
the $L_n$ we have that the area of $L$ is greater than sum of the area
of $L_n$.

Finally, let $H_{-2m}$ be the polygon obtained by shifting the Newton
polygon $H$ vertically by $-2m$. By construction, the index $u_n$ has
been chosen such that the vertical distance between $\mathscr L$ and the
point of abscissa $u_n$ of $H_{-m}$ is at most $m$. Thus the points
where $\mathscr L$ exits $V_n$ on the right is above $H_{-2m}$. Moreover, at
abscissa $\ell_n$, the line $\mathscr L$ is above $H_{-m}$, so that in
the vertical band $V_n$, the line $\mathscr L$ lies entirely above
$H_{-2m}$. Similarly, we can prove that $\mathscr L'$ is above $H_{-2m}$ since it is above
$\mathscr L$ at abscissa $u_n$, and at abscissa $\ell_n$, its vertical
distance to $\mathscr L$ is $m_{\ell_n} \leq m$. Thus $U_n$ and $L_n$
are contained in the band between $H$ and $H_{-2m}$. This implies that
$U$ (resp. $L$) has an area less than $2dm$.

Gathering all the geometrical constraints, and letting
$|U|$, $|L|$, $|U_n|$, $|L_n|$ be the areas of the corresponding regions, we have:
$$\sum_{\substack{0\leq n <N~|~\delta_n\geq 1}} \frac{\delta_n m}8 \leq
\sum_{\substack{0\leq n <N~|~\delta_n\geq 1}} |U_n|+|L_n| \leq |U| + |L|
\leq 4dm$$
such that $\sum_{0\leq n <N} \delta_n \leq 32d$. And for $\delta_n\geq1$
we have $\delta_n+1\leq 2\delta_n$, which implies $\sum_{0\leq n <N,
\delta_n\geq1}
(\delta_n+1) \leq 64d$.

To conclude the proof of the inequality (ii) of Lemma~\ref{lem:rings},
we need to take also into account the case where $\delta_n = 0$. This
case happens when $\ell_n=u_n$. The values $\ell_n$ such that
$\delta_n=0$ are all distinct, such that $\sum_{0\leq n<N, \delta_n=0}
(\delta_n+1) \leq d+1$.

%
%

\subsubsection{Bound on the complexity of Algorithm~\ref{alg:rings}}
 \label{sec:ringcomplexity}

We prove here that the number of bit operations of
Algorithm~\ref{alg:rings} is in $\widetilde O(d(m+\log\tau))$.
Let $N$ be the number of iterations of the {\bf while} loop.

First, the lists $S^0$ and $S^1$ can be computed with $O(d)$
arithmetic operations. Second, since the indices 
$\ell$ and $u$ are non-decreasing, computing the sequences
$(\ell_n)_{n\leq N}$ and $(u_n)_{n\leq N}$ amount to 
$O(N+d)$ comparisons.
Third, the terms of the sequence $(r_n)_{n\leq N}$
can be computed using
\sloppy Arithmetic-Geometric means \cite{BRENT1976151}. This method is quasi-linear in
required number of correct significant bits. The $h_i$ have
a magnitude bounded by $\tau$. Thus, the variable $s$ has a
magnitude in $O(\tau)$ and it is sufficient to perform the arithmetic
operations and the exponentiation with $O(m + \log d + \log \tau)$ correct significant digits
in order to get the required precision for a term $r_n$.

Finally, using inequality (ii) of Lemma~\ref{lem:rings},
$N$ is in $O(d)$ which
concludes the proof.

\section{Piecewise approximation over angular sectors}
\label{section_approximation}
In the previous section, we computed a first piecewise approximation of
$f$ by polynomials $z^{\ell_n} f_{\ell_n,u_n}(z)$ over a partition of the complex
plane in concentric rings $R_n$.
Even though the total number of monomials of those polynomials is linear in $d$, it is possible that a
single approximating polynomial $f_{\ell_n, u_n}$ has a degree $d$.
Thus we need to further truncate the polynomials so that
subsequent evaluations or root finding methods are faster. This is done
by dividing uniformly each ring $R_n$ in angular sectors $A_{n,k}$, and by
computing a Taylor approximation of $f_{\ell_n,u_n}$ in each sector. More
precisely, let
$\gamma_{n,k}$ be the center and $|\rho_{n,k}|$ be the radius of a disk that
contains $A_{n,k}$. Then we
approximate $f_{\ell_n,u_n}$ by computing the first $4m+1$ coefficients
of the polynomial $f_{\ell_n,u_n}$ composed with the
polynomial $\gamma_{n,k} + \rho_{n,k}t$.
We denote by $g_{n,k}$ the resulting polynomial. The main result of this
section is summarized in the following Lemma.

\begin{lemma}
  \label{lem:approximations}
  Given a polynomial $f_{\ell_n,u_n}$ satisfying the conditions of
  Lemma~\ref{lem:rings}, it is possible to compute in $\widetilde
  O(\delta_n (m + \log \tau))$ bit operations the parameters of $K$ affine changes of
  variable $Z = \gamma_{n,k} + \rho_{n,k}T$ and $K$ polynomials $g_{n,k}$ of degree $\min(\delta_n, 4m)$
  such that:
  \begin{gather}
  \tag{i} R_n\text{ is included in the union of the disks }D(\gamma_{n,k},|\rho_{n,k}|), \\
  \tag{ii} \|f_{\ell_n,u_n}(\gamma_{n,k} + \rho_{n,k}T) - g_{n,k}(T)\|_1 \leq (\delta_n+1) 2^{-m} 
  \hat f(r_n) r_n^{-\ell_n}.
  \end{gather}
\end{lemma}

We use Algorithm~\ref{alg:approximations} to compute those
approximations. We describe it and provide its complexity analysis in
Section~\ref{sec:tayloralgo}, and we prove the inequalities (i) and (ii)
in Section~\ref{sec:taylorineq}.


%
%
%
%
%
%

\subsection{Algorithm and complexity analysis}
\label{sec:tayloralgo}
In Algorithm~\ref{alg:approximations} we focus on one ring $R_n$
defined by two circles of radii $r_n$ and $r_{n+1}$. The algorithm is a
variant of the algorithm used in a previous work where the rings were
given by explicit formula \cite[\S 3]{moroz2022new}. The main differences are
that we need to handle arbitrary rings, and we need to guarantee a
different error bound for the approximation polynomials $g_{n,k}$.

\begin{algorithm}
  \DontPrintSemicolon
  \KwInput{%
    \begin{minipage}[t]{7cm}%
      \begin{itemize}
        \item[$m$:] a positive integer
        \item[$f$:] input polynomial
        \item[$r_n, r_{n+1}$:] two consecutive radii surrounding a ring
        \item[$\ell_n$:] lower index of the dominant monomials
        \item[$u_n$:] upper index of the dominant monomials
      \end{itemize}         
    \end{minipage}
  }
  \KwOutput{Satisfying inequalities (i),(ii) of Lemma~\ref{lem:approximations}
    \begin{minipage}[t]{0.8\linewidth}%
      \begin{itemize}
        \item[$(\gamma_{n,k},\rho_{n,k})_{0\leq k < K_n}$:] change of variable parameters: $z = \gamma_{n,k} + \rho_n t$
        \item[$(g_{n,k})_{0\leq k < K_n}$:] list of approximation
          polynomials
      \end{itemize}         
    \end{minipage}
  }
  \vspace{1em}
  \tcc{Computes truncated Taylor shifts using a\\baby-step giant step approach}
  \vspace{1em}
  \tcc{A. Parameters for the change of variable\\
  \phantom{A.} $z=(\gamma+\rho t)e^{i2\pi\frac k K}$}
  $\gamma, \rho, \gets 
  \frac{r_{n+1}+r_n}{2}, \frac 3 2 \frac{r_{n+1}-r_n} 2$ \;
  $\delta \gets u_n - \ell_n$\;
  $K \gets \lceil 2 \pi \frac{\gamma}{\rho} \rceil$\tcc*{$K \leq 3 \delta / m + 2$}
  $m' \gets \lfloor \frac{\delta}{K} \rfloor$\;
  \vspace{1em}

  \tcc{B. Baby steps.\\ 
       \phantom{B.} Compute $(\gamma+\rho T)^k \mod T^{4m+1}$, normalized}
  $q_0(T) \gets \frac{\gamma + \rho T}{\gamma+\rho}$\;
  \For{$k$ from $0$ to $K-1$}{
    $q_{k+1}(T) \gets q_{k}(T) \cdot \frac{\gamma + \rho T}{\gamma +
    \rho} \mod T^{4m+1}$\;
  }
  \vspace{1em}
  \tcc{C. Giant steps}
  \tcc{C.1 Gather monomials with same indices modulo K}
  $M \gets m' \frac{\hat f(\gamma)}{\gamma^{\ell_n}}(1+\frac \rho \gamma)^{\delta}$\;
  \For{$k$ from $0$ to $K-1$}{
    $p_k(Z) \gets \frac 1 M \sum_{j=0}^{m'} f_{\ell_n + k + j
    K}(\gamma+\rho)^{k+jK} Z^j$\;
  }
  \tcc{C.2 Fast composition with absolute error $2^{-4m}$}
  \For{$k$ from $0$ to $K-1$}{
    $r_k(T) \gets \sum_{j=0}^{4m} r_{k,j} T^j = p_k\left(q_{K}(T)\right)
    \cdot q_k(T) \mod T^{4m+1}$\;
  }
  \vspace{1em}
  \tcc{D. Spread approximations to all angular sectors}
  \For{$j$ from $0$ to $\min(\delta, 4m)$}{
    $r_j^T(W) \gets \sum_{k=0}^{K-1} r_{k,j} W^k$\; 
    $g_{0,j}, \ldots, g_{K-1,j} \gets r^T(e^{i2\pi\frac k {K}})$ for $0
    \leq k < K$ with FFT\;
  }
  \Return $(\gamma  e^{i2\pi\frac k {K}}, \rho e^{i2\pi\frac k {K}})$ for $0\leq k < K$\\ 
  \phantom{\Return}$\sum_{j=0}^{4m} M g_{k,j} T^j$ for $0\leq k < K$\;
  \caption{Piecewise polynomial approximation}
  \label{alg:approximations}
\end{algorithm}

We start by computing the parameters of the disks covering the ring
$R_n$, where $\gamma$ is the middle point between $r_n$ and $r_{n+1}$, whereas
$\rho$ is $3/2$ the radius of the interval $[r_n,r_{n+1}]$. With those
parameters, we can check that the disk $D(\gamma,\rho)$ contains an
angular sector of $R_n$ of angle $2 \rho / \gamma$
(Lemma~\ref{lem:disk}). The number of disks necessary to cover $R_n$ is thus bounded by $K \leq 3
(\delta / m) +2$, which is in $O(\max(\delta,m)/m)$. Then we want to
compute the first $4m+1$ coefficients of the polynomials in $t$ obtained
after the change of variable $z=(\gamma + \rho t)e^{i2\pi k/K}$ for
all $0\leq k < K$ in the polynomial $f_{\ell_n,u_n}$;
next we show how to do it
with bit complexity
$\widetilde O(\delta m)$ using a baby-step giant-step or rectangular
splitting approach \cite{brent2010modern}.


For that, remark that when computing the fast Fourier
transforms, the $e^{2\pi j/K}$ are equal for indices that are equal
modulo $K$. Thus for a given $k$ between $0\leq k < K$, we can gather
the terms of index $k$ modulo $K$ in the polynomial
$r_k = \sum_{j=0}^{\delta/K} h_{k+jK}$ before computing the fast Fourier
transform. This will allow us to reorder our computations 
to compute the $g_{n,k}$ in
$\widetilde O(\delta m)$ bit operations.

More precisely, in step \texttt{B.}
(the baby-step) 
we compute the powers of the
polynomials $\gamma + \rho T$ modulo $T^{4m+1}$ up to the exponent $K$
which is in $O(\delta/m+1)$. At each step of the loop, normalizing
the linear factor can be done in $\widetilde O(m+\log \tau)$ bit
operations using floating-point arithmetic since $\gamma$ and $\rho$
have a magnitude in $O(2^\tau)$.  Then each polynomial multiplication can
be done in
$\widetilde O(\min(\delta,m) m)$ bit operations. Thus all the polynomials $q_k$ can
be computed in $\widetilde O(\min(\delta,m) Km + K(m+\log \tau)))$, or $\widetilde
O(\delta m + K \log \tau)$ since $Km$ is in $O(\max(\delta,m))$.

Then in step \texttt{C.} 
(the giant-step) 
we compute the coefficients of the polynomials
$r_k$ of $q_k \cdot (p_k \circ q_K) \mod T^{4m+1}$ by using a
combination of fast Taylor shift (\cite[Theorem~8.4]{schonhage1982fundamental}) and fast
composition (\cite[Theorem~2.2]{Rtcs86}).
This can be done with an absolute error on the coefficients less than $2^{-4m}$
in $\widetilde O(\min(\deg(p_k)\deg(q_k),m)m)$ bit
operations.
Since the degree of $q_k$ is $K$ and the degree of $p_k$ is $m' \leq
\delta /K$, the degree $\deg(p_k)\deg(q_k)$ of $p_k\circ q_k$ is in
$O(\delta)$ and all the $r_k$ can be computed in $\widetilde
O(\min(\delta,m)Km)$ bit operations. Finally, $Km$ is in
$O(\max(\delta,m)$, such that step \texttt{C.2} can be done in
$\widetilde O(\delta m)$ operations. Moreover, as in step \texttt{B.}, normalizing the
coefficients of $p_k$ in step \texttt{C.1} costs $\widetilde
O(\delta(m+\log \tau))$.

Finally in step \texttt{D.} we compute the polynomials $g_{n,k}$ using $\min(\delta, 4m)$ fast Fourier
transforms with an absolute error less than $2^{-4m}$
(\cite{schonhage1982fundamental}). This can be done
in $\widetilde O(\min(\delta,m) Km)$ bit operations, that is in
$\widetilde O(\delta m)$.

Thus, the total number of bit operations of
Algorithm~\ref{alg:approximations} is in $\widetilde O(\delta (m+\log \tau))$.

\subsection{Correctness}
\label{sec:taylorineq}
We prove here 
inequalities (i) and (ii) of
Lemma~\ref{lem:approximations}.

\vspace*{-0.15cm}
\subsubsection{Ring cover}
For the inequality (i) of Lemma~\ref{lem:approximations}, the ring $R_n$ is contained in the union of the
$K_n$ disks $D(\gamma_{n,k},|\rho_{n,k}|)$ if the following Lemma holds. 


\begin{lemma}
  \label{lem:disk}
The disk $D(\gamma,\rho)$ contains an angular sector of $R_n$ of angle
$\frac \rho \gamma$. Furthermore, $\frac \gamma \rho \leq 1 + 3 \frac \delta m$.
\end{lemma}

\begin{proof}
The angular sector covered by $D(\gamma, \rho)$ has an angle $2\beta$,
where $\beta$ is the angle between two sides of lengths $\gamma$ and
$r_{n+1}$ in a triangle with sides of lengths $\gamma$,
$r_{n+1}$ and $\rho$. We have $\beta^2 \geq 2(1-\cos(\beta))$. Then
using arguments from the triangle geometry to bound $\cos(\beta)$, we
deduce $2\beta > (2/3)\sqrt{5/2} \rho/\gamma > \rho / \gamma$. 

Then for the bound on $\frac \gamma \rho$, we remark that it is equal to
$(2/3)(r_{n+1}+r_n)/(r_{n+1}-r_n) = (2/3)(1 + 1/(r_{n+1}/r_n-1))$. Moreover,
$r_{n+1}/r_n \geq 2^{m/(2\delta)} \geq 1 + (\log(2)/2) m/\delta$. With
$2/3 < 1$ and $4/(3\log(2)) < 3$, this allows us to conclude the proof.
\end{proof}

          
\vspace*{-0.15cm}
\subsubsection{Approximation bound}
To prove the bound on the approximation inequality (ii) of
Lemma~\ref{lem:approximations}, we need to bound two errors: the error
appearing while truncating at order $4m$ the polynomials
$f_{\ell_n,u_n}(\gamma + \rho T)$, and the error appearing on the
coefficients while computing with an absolute error in $2^{-4m}$ in
Algorithm~\ref{alg:approximations}.

First we bound the error coming from the error on the coefficients.  In
steps~\texttt{B} and~\texttt{C.1}, we normalize the polynomials such
that the sum of the absolute value of their coefficients is less than
$1$. Thus, at the end, the absolute error on each coefficient of
$g_{n,k}$ is less than $(\delta+1)2^{-4m}M$. We show that it is bounded
by $(\delta+1)2^{-m-1} \hat f(r_n)/r_n^{\ell_n}$ using the following
bound on
$M=m'(\hat f(\gamma)/\gamma^{\ell})(1+\rho/\gamma)^{\delta}$.

\begin{lemma}
  \label{lem:boundM}
  For all $m\geq1$:
  $$\|g_{n,k}\|_1 \leq M \leq (\delta+1)2^{3m-1}\hat f(r_n)/r_n^{\ell_n}$$
\end{lemma}
\begin{proof}
Let $j$ be
the index such that $\hat f(\gamma) = |f_j| \gamma^j$. We have
$|f_j|\gamma^{j-\ell_n} =
|f_j|r_n^{j-\ell_n}(\gamma/r_n)^{j-\ell_n}\leq(\hat
f(r_n)/r_n^{\ell_n})(r_{n+1}/r_n)^\delta$. With $(r_{n+1}/r_n)^\delta
\leq 2^m$, this leads to $M \leq m'(\hat
f(r_n)/r_n^{\ell_n})2^m(1+\rho/\gamma)^\delta$. Then, we can notice that
$(1+\rho/\gamma)^\delta \leq 2^{(1/\log(2))(\delta\rho/\gamma)}$.
\sloppy Moreover, $\rho/\gamma = (3/2) \tanh((s_{n+1}-s_n)\log(2)/2)$. Thus
$\rho/\gamma\leq 3\log(2)(s_{n+1}-s_n)/4\leq (3\log(2)/4)(m/\delta)$.
This leads to $(1+\rho/\gamma)^\delta \leq 2^{3m/4}$. This also leads to
$K\geq (8\pi/(3\log(2)))(\delta/m)$, such that $m' \leq m$. Finally,
this implies that $M$ is bounded by $m2^{7m/4}\hat f(r_n)/r_n^{\ell_n}$.
Moreover, for all $m\geq 1$, we have $m2^{7m/4} \leq 2^{3m-1}$.
\end{proof}


Then the bound on the error coming from the Taylor expansion is obtained by bounding
for a given $j$ the coefficients of the polynomial $|f_j|(\gamma + \rho
T)^{j-\ell_n}$ of degree greater than $4m$. The $k$-th coefficient
of this polynomial is $|f_j|\gamma^{j-\ell_n}\binom{j-\ell_n}{k}
(\rho/\gamma)^k \leq |f_j|\gamma^{j-\ell_n} ((e \delta/k)
(\rho/\gamma))^k$. With the bound on $\rho/\gamma$ above, this becomes
less than $|f_j|\gamma^{j-\ell_n}((3e\log(2)/4)(m/k))^k$. As
previously, we can bound $ |f_j|\gamma^{j-\ell_n}$ by $2^m\hat
f(r_n)/r_n^{\ell_n}$. The second factor can be bounded for all $k\geq
4m$ by $(3e\log(2)m/16)^k\leq 1/2^k$. Thus the sum of all the
coefficients of degree $4m+1$ or more is bounded by $(\hat
2^mf(r_n)/r_n^{\ell_n})2^{-4m}$. Adding the errors for $j$ from $\ell_n$
to $u_n$, we get a bound for the error on the remainder of
$(\delta+1)\hat f(r_n)/r_n^{\ell_n} 2^{-m-1}$.

Summing the two errors lead to the required bound for the
inequality~(ii) of Lemma~\ref{lem:approximations}.

\section{Evaluation}
\label{section_evaluation}

As a consequence of Lemma~\ref{lem:rings} and~\ref{lem:approximations},
the polynomials $g_{n,k}$ returned by Algorithm~\ref{alg:approximations}
form a piecewise polynomial approximation over the angular sectors
$A_{n,k}$ defined by the intersection between the ring $R(r_n, r_{n+1})$ and the disk
$D(\gamma_{n,k}, |\rho_{n,k}|)$. This is formalized in the following
corollary.

\begin{corollary}
  Let $z$ be a complex point. Then there exists $n$ and $k$ such that
  $r_n\leq |z|\leq r_{n+1}$ and $z\in D(\gamma_{n,k}, |\rho_{n,k}|)$.
  Moreover, letting $t=(z-\gamma_{n,k})/\rho_{n,k}$, the polynomial $g_{n,k}$ satisfies:
  $$|f(z) - z^{\ell_n}g_{n,k}(t)| \leq (d+1)2^{-m}\hat f(z)$$
\end{corollary}

This leads to a straightforward method to evaluate $f$ on one point $z$ with an error in $O(d2^{-m}\widetilde f(|z|)$.  First
find the angular sector $A_{n,k}$ it belongs to, then 
compute $t =
(z-\gamma_{n,k})/\rho_{n,k}$ and then
evaluate
$z^{\ell_n}g_{n,k}(t)$. To evaluate $f$ at $d$ points, 
we use fast multipoint evaluation
\cite{kobel2013fast,Hrr08} to evaluate the
$g_{n,k}$.

\vspace*{-0.2cm}
\paragraph{Complexity analysis}

Assume that $z$ has a magnitude less than $2^\mu$.
Finding the angular sector $A_{n,k}$ can be done with a binary
search on the most significant bits for a total number of bit operations
in $\widetilde O(\log d + m+\log \mu)$. Then the change of variable
$t = (z-\gamma_{n,k})/\rho_{n,k}$ can be done within the same complexity.

Finally, we can evaluate $g_{n,k}(t)$ in $\widetilde O(m^2+\log \tau)$ bit
operations with a classical Hörner scheme, and we can evaluate $z^{\ell_n} = e^{\ell_n\log z}$ in $\widetilde O((m+\log
\mu + \log d)$ bit operations using arithmetico-geometric means for
the logarithm and the exponential.  Thus the total number of
bit operations for one evaluation is in $\widetilde O(m^2 + \log d + \log
\tau + \log\mu)$.

Furthermore, to evaluate $f$ on a set of $d$ points of magnitude less than $\log
\mu$, after gathering the points by angular sectors, 
using a
fast multipoint evaluation algorithm \cite{kobel2013fast,Hrr08} to evaluate
$g_{n,k}$ on each sector, leads to a number of bit operations in $\widetilde
O(d(m+\log \mu + \log \tau))$.

\vspace*{-0.2cm}
\paragraph{Computation of the error bound}
We can also compute the error bound efficiently with the explicit
formula $(d+1)2^{-m}\hat f(z)$. Evaluating $\hat f(|z|)$ requires first to
find the index $j$ such that $\hat f(|z|) = |f_j||z|^j$. This amounts to
find the vertex of the Newton polygon $H$ belonging to the line of slope $\log
|z|$ tangent to $H$. 
If the slopes of the edges of $H$ are computed beforehand
this can be done with a binary search $\widetilde O(\log d + m + \log \mu)$ bit operations. Then evaluating the
formula can be done in $\widetilde O(m + \log \mu + \log d + \log \tau)$ bit
operations. Thus, the error bound can be evaluated within the same
number of operations as the approximate evaluation of $f$ itself.

%
%

\section{Root finding}
\label{section_isolation}

%

Our algorithm to isolate the roots of $f$ consists in computing first a piecewise
approximation of $f$, then finding the
approximations $\dot \zeta$ to the roots of each approximate  polynomial, and finally, for each approximate root  $\dot \zeta$,
computing a disk centered on $\dot \zeta$
that contains a root $\zeta$ of $f$ if $\zeta$ is well-conditioned.
It is described in more details in Algorithm~\ref{alg:rootfinding}.

The main lemma that guarantees that our algorithm approximates correctly
all the well-conditioned roots of $f$ is the following.

\begin{lemma}
  \label{lem:rootfinding}
  With the notations of Algorithm~\ref{alg:rootfinding}:
  \begin{itemize}
    \item[(i)] each $D(\dot\zeta,r)\in\dot Z$ contains a unique root $\zeta$ of $f$
      and the Newton algorithm starting at any point in the disk
      converges toward $\zeta$
    \item[(ii)] each root $\zeta$ of $f$ such that
      $2\log(\cond(f,\zeta))+3\log_2(d+1) + 11<m$ is contained in a disk $D(\dot\zeta,r) \in
      \dot F$
  \end{itemize}
\end{lemma}

To ensure that we find all the root $\zeta$ such that
$\cond(f,\zeta)\leq 2^m$, it suffices to change $m$ by $2m +
3\log_2(d+1) +11$ before
running the algorithm.
%
We prove inequality (i) in Section~\ref{sec:rootunique} and inequality
(ii) in Section~\ref{sec:rootterminates}

For the complexity of the algorithm, the dominating parts are: the computation
of the piecewise approximation in step \texttt{A}, the $O(d/m)$ root approximations in step
\texttt{B}, and the approximate evaluations $\dot F$ and $\dot F'$ in step \texttt{C}. The
computation of the piecewise approximation and the evaluation can be
done in $\widetilde O(d(m+\log \tau))$ bit operations, and each root approximation can
be done in $\widetilde O(m^2)$ bit operations, which prove the
complexity stated in Theorem~\ref{thm:rootfinding}.

%
%

\begin{algorithm}
  \DontPrintSemicolon
  \KwInput{$m$ a positive integer and a polynomial $f$ of degree $d$}
  \KwOutput{List of isolating disks of the roots $\zeta$ of $f$ such that
  $2\log_2(\cond(f,\zeta))+3\log_2(d+1) +11 \leq m$}
    \tcc{A. Compute the piecewise approximation}
    $A_m(f) \gets$ the piecewise approximation of $f$\;
    \vspace{1em}
    \tcc{B. Compute the approximate roots}
    $Z, \dot Z \gets \{\}$\;
    \For{$(g, \gamma, \rho)$ in $A_m(f)$ \label{num:startfor}}{
      $\dot g \gets$ Approximate factorization of $g$ such that \label{line:factor}\\
      \phantom{$\dot g\gets$} $\|g - \dot g\|_1 \leq 2^{-4m} \|g\|_1$\;
      $\dot Z \gets \dot Z \cup \{\gamma + \rho \dot t \mid \dot t$ root of $\dot g \}$\;
      \label{num:endfor}
    }
    \vspace{1em}
    \tcc{C. Evaluate $f$ and $f'$ on the approximate roots}
    $\varepsilon(z) \gets z \mapsto (d+1) 2^{-m}\hat f(|z|)$\;
    $\varepsilon'(z) \gets z \mapsto d^2 2^{-m}\hat f(|z|)/|z|$\;
    $\dot F \gets$ Approximate evaluations of $f$ on $\dot Z$ with error $\varepsilon$\;
    $\dot F' \gets$ Approximate evaluations of $f'$ on $\dot Z$ with
    error $\varepsilon'$\;
    \vspace{1em}
    \tcc{D. Compute the isolating disks}
    \For{$(\dot \zeta, \dot f, \dot f') \in \dot Z \times \dot F \times \dot
    F'$ s. t. $\dot f\simeq f(\dot \zeta)$ and $\dot f' \simeq f'(\dot
  \zeta)$}{
      $r \gets 2 \frac{|\dot f| + \varepsilon(\dot \zeta)}{|\dot f'| -
      \varepsilon'(\dot \zeta)}$\;
      $K \gets 2 d^3 \frac{\hat f(|\dot \zeta|)}{|\dot \zeta|^2(|\dot
      f'| - \varepsilon'(\dot \zeta))}$ \label{line:K}\;
      \If{$4r < |\dot\zeta|(2^{1/d}-1)$ and $5 r K < 1$ \label{line:criterion}}{
        $Z \gets Z \cup \{D(\dot \zeta, r)\}$\;
      }
    }
    \Return $L$\;
  \caption{Root isolation}
  \label{alg:rootfinding}
\end{algorithm}

\subsection{Guarantee that the each returned disk isolate a root of $f$}
\label{sec:rootunique}
Given a root $\zeta$ of $f$, the Newton bassin of $\zeta$ is the set of
initial points such that the Newton method converges toward $\zeta$.
Using a bound on the second derivative of $f$, the Kantorovich's theory
allows us to give a bound on the radius of a disk centered at $\dot \zeta$
that is included in the Newton bassin of $\zeta$ \cite[\S 3.2]{Dbook06},\cite[Lemma 1]{moroz2022new}. Namely, let $r >
2|f(\dot \zeta)/f'(\dot \zeta)|$ and $K > |f''(y)/f'(\dot z)|$ for all
$y \in D(\dot \zeta, 4r)$. If $5rK \leq 1$ then, $f$ has a unique root
$\zeta$ in $D(\dot \zeta, r)$, and for all $z \in D(\dot \zeta,r)$, the
Newton sequence starting from $z$ converges to $\zeta$.


In our case, remark that for all $z$ such that $|z| \leq |\dot \zeta|
2^{1/d}$, we have $|f''(z)| \leq d^2\widetilde
f_{2,d}(|z|)\leq 2 d^2\widetilde f_{2,d}(|\dot\zeta|)\leq 2
d^2\widetilde f(|\dot\zeta|)/|\dot\zeta|^2\leq 2d^3\hat f(|\dot\zeta|)/|\dot\zeta|^2$.
%
This ensures that the number $K$ computed on line~\ref{line:K} of
Algorithm~\ref{alg:rootfinding} is an upper bound on $\max_{z\in
D(\dot\zeta,r)} |f''(z)/f'(\dot\zeta)|$. Thus using the Kantorovich's
theory, this ensures that 
when the criterion on 
line~\ref{line:criterion} of Algorithm~\ref{alg:rootfinding} is
satisfied, the returned disk isolates a root of $f$. Moreover, if
two such disks intersect, then the root is in their intersection since
the Newton method starting from any point in their
intersection will converge toward a unique root in the
their intersection.

\subsection{Guarantee that each well-conditioned root is contained in a
returned disk}
\label{sec:rootterminates}
Given a root $\zeta$ of $f$, the proof of inequality (ii) of Lemma~\ref{lem:rootfinding} is done in
two steps. First we show that $\dot F$ contains a root close enough to
$\zeta$. Then we show that the Kantorovich criterion on
line~\ref{line:criterion} is satisfied.

In this section, we assume that
$\zeta$ belongs to an angular sector $A_{n,k}$,
and we let $g$ be the corresponding approximate polynomial, and $\gamma$ and $\rho$ be the parameters for the change of
variable $z = \gamma + \rho t$. Then $\dot g$ denotes the approximate
polynomial of $g$ obtained by factorization in line~\ref{line:factor} of
Algorithm~\ref{alg:rootfinding}. 
Finally, let $\dot \zeta$ be the closest point to $\zeta$ such that
$\dot g((\dot \zeta - \gamma)/\rho) = 0$, and let
$\theta = (\zeta - \gamma)/\rho$ and  $\dot \theta = (\dot
\zeta - \gamma)/\rho$, such that $f(\zeta) = 0$ and $\dot g(\dot \theta) = 0$.
       
\subsubsection{Bound on the distance of the approximated roots}
\label{sec:rootbound}
                                 
In this section we compute an upper bound on the distance
between $\zeta$ and $\dot \zeta$. First it is known (\cite[Theorem
6.4e]{Hbook74}, \cite[Theorem 9]{BAna00}) that for a given point $\theta$,                                                             
the distance between $\theta$ and the closest root $\dot \theta$ of
$\dot g$ is bounded by $\deg(\dot g) |\dot g(\theta)/\dot g'(\theta)|
\leq d |\dot g(\theta)/\dot g'(\theta)|$.
%
Then by construction, $\dot g$ satisfies 
  $\left|g\left(\theta\right) - \dot g(\theta)\right| \leq
  2^{-4m}\|g\|_1$ and
  $\left|g'\left(\theta\right) - \dot g'(\theta)\right| \leq d
  2^{-4m}\|g\|_1$.

Using Lemma~\ref{lem:boundM} we have $\|g\|_1 \leq
(\delta+1)2^{-3m-1}\hat f(\|\zeta|)/|\zeta|^{\ell_n}$.
With the triangular inequality, combining the
inequality on $|g(\theta) - \dot g(\theta)|$  
with the inequalities (iii) of
Lemma~\ref{lem:rings} and (ii) of Lemma~\ref{lem:approximations}, leads to $\left|\dot g(\theta)\right| \leq
2(d+1)2^{-m} \widetilde f(|\zeta|)/|\zeta|^{\ell_n}$.

For the derivative, we can show with similar arguments that
$|\rho f'_{\ell_n,u_n}(\gamma+\rho\theta) - \dot g'(\zeta)|\leq
2(d+1)^2\widetilde f(|\zeta|)/|\zeta|^{\ell_n}$. We can also deduce from
inequality (iii) of Lemma~\ref{lem:rings} that $\rho|f'(\zeta)
- \zeta^{\ell_n} f'_{\ell_n,u_n}(\zeta)| \leq \rho d^2 \widetilde
f(|\zeta|)/|\zeta|$. Finally, remark that $|\zeta|/\rho \leq
r_{n+1}/((3/4)(r_{n+1}-r_n))\leq (4/3)(1/(1-r_n/r_{n+1})) \leq
(4/3)(1/(1-2^{-m/2d}))\leq (4/3)(2d/m+1)$.
This implies that for $m\geq 4$ and $d\geq1$ we have $|\zeta|/\rho \leq
d+1$ and $
  \left|\dot g'\left(\theta\right) -
  \rho{f'(\zeta)}/{|\zeta|^\ell_n}\right| \leq 3 \rho (d+1)^32^{-m}
  {\widetilde f(|\zeta|)}/{|\zeta| |\zeta|^{\ell_n}}$.

%
%
%

Combining all the previous inequalities, this implies that for $m \geq
\log_2(\cond(f,\zeta))+3\log_2(d+1)+4$ we have:
  $$\frac{|\zeta - \dot \zeta|}{|\zeta|} \leq 4 (d+1) 2^{-m} \cond(f,\zeta)$$

\subsubsection{Guarantee that each well-conditioned root is returned}
Finally, to prove the correctness of Algorithm~\ref{alg:rootfinding}, it
remains to prove that if 
$2\cond(f,\zeta) + 3\log_2(d+1) + 11\leq m$, then 
the criterion on line~\ref{line:criterion} applied to $\dot \zeta$ evaluates to true.
Although the proof is technical, the main idea is that the variable
$r$ is roughly proportional to $|\zeta - \dot\zeta|$. Thus, as soon as
$|\zeta - \dot\zeta|$ is small enough, rK will become small enough to validate the
criterion. 

First, if $m > \cond(f,\zeta) + 3 \log_2(d+1) +4$, the bound on
$|\zeta-\dot\zeta|$ ensures that
$|\zeta|2^{-1/d} \leq |\dot \zeta| \leq |\zeta|2^{1/d}$. In particular this
implies $\widetilde f(|\dot \zeta|) \leq 2 \widetilde f(|\zeta|).$
Moreover for all points of module less than $|\zeta|2^{1/d}$ we have
$|f'(z)|\leq 2d\widetilde f(|\zeta|)/|\zeta|$ and $|f''(z)| \leq 2d^2
\widetilde f(|\zeta|)/|\zeta|^2$.

Thus we can bound
$|f(\dot \zeta)|$  and $|f'(\dot\zeta)|$ using a Taylor expansion around
$\zeta$ with $|f(\dot\zeta)| \leq |\zeta-\dot\zeta||f'(\zeta)| + d^2
|\zeta-\dot\zeta|^2\hat f(|\zeta|)/|\zeta|^2$ and 
$|f'(\dot\zeta)| \geq |f'(\zeta)| - 2d^2 |\zeta-\dot\zeta|\widetilde
f(|\zeta|)/|\zeta|^2$.


Letting $C = 6(d+1)\cond(f,\zeta)$, this leads to $|\dot f| +
\varepsilon(\dot\zeta) \leq |\zeta||f'(\zeta)| (2^{-m}C + d2^{-2m}C^3)$
and $|\dot f'| -\varepsilon'(\dot\zeta) \geq |f'(\zeta)|\left(1 -
d2^{-m}C^2 - d2^{-m}C\right)$.


Thus as soon as $d2^{-m}C(C+1) < 1/2$, we have $r \leq
6\cdot2^{-m}C|\zeta|$ and $K \leq 8d^3 \widetilde
f(|\zeta|)/(|\zeta|^2|f'(\zeta)|)$,
such that $5rK \leq 40 d 2^{-m}C^2$, which is below $1$ when 
$40d2^{-m}C^2<1$, and $r/|\dot\zeta| \leq 6\cdot 2^{-m}C2^{1/d}$, which is
below $\log(2)/d < 2^{1/d}-1$ when $18d\cdot2^{-m}C < \log(2)$. All those
constraints are satisfied for $1500 (d+1)^3 \cond^2(f,\zeta) < 2^m$,
which is satisfied for $m > 2\log_2(\cond(f,\zeta)) + 3\log_2(d+1) +
11$.

\section{Benchmarks and applications}
\label{section_benchmarks}

    \newcommand {\R}   {\mathbb R}
	\newcommand {\Q}   {\mathbb Q}

\newcommand{\intbox}{%
	  {\,\setlength{\unitlength}{.33mm}\framebox(4,7){}\,}}

\newcommand{\ball}[2]{ \intbox_{#1}{#2} }
\newcommand {\aprec} {p}
\newcommand{\ballp}[1]{ \ball{\aprec}{#1} }
	  
\newcommand{\machine}{\texttt{Intel(R) Core(TM) i7-8700 CPU @ 3.20GHz}\xspace}
\newcommand{\clang}{\texttt{C}\xspace}
\newcommand{\pwpoly}{\texttt{PWPoly}\xspace}
\newcommand{\mpsolve}{\texttt{MPSolve}\xspace}
\newcommand{\hcroots}{\texttt{HCRoots}\xspace}
\newcommand{\pwroots}{\texttt{PWRoots}\xspace}
\newcommand{\pweval}{\texttt{PWEval}\xspace}
\newcommand{\pwpolyurl}{https://gitlab.inria.fr/gamble/pwpoly}
\newcommand{\arb}{\texttt{Arb}\xspace}
\newcommand{\arburl}{https://arblib.org}
\newcommand{\mpsolveurl}{https://numpi.dm.unipi.it/scientific-computing-libraries/mpsolve/}
\newcommand{\mpfr}{\texttt{mpfr}\xspace}

%
%
%
The approaches presented here 
are not only of theoretical interest, and can be applied to solve problems involving well-conditioned polynomials of large degrees with coefficients of different orders of magnitude, for instance truncated series
expansions of
Gaussian analytic functions, giving polynomials with repulsion at the first order,
or truncated series
expansions of hypergeometric functions.
%
We
demonstrate their practical efficiency for
several families of well-conditioned
polynomials
that are random polynomials associated
with hyperbolic, elliptic and flat
distributions, as well as more structured polynomials
as resultant of random bivariate polynomials
appearing when solving multi-variate systems
with elimination.

    
For $d>=0$ 
we define the hyperbolic,
elliptic and flat bases as
$\left( \mathcal{H}_i \right)_{0\leq i \leq d}:=\left( 1 \right)_{0\leq i \leq d}$,
$\left( \mathcal{E}_i \right)_{0\leq i \leq d}:=\left( \sqrt{d\choose i} \right)_{0\leq i \leq d}$
and 
$\left( \mathcal{F}_i \right)_{0\leq i \leq d}:=\left( 1/\sqrt{i!} \right)_{0\leq i \leq d}$.
A polynomial $f$ can be decomposed as
$$f(z)=\sum_{i=0}^d a_i\mathcal{H}_iz^i
   =\sum_{i=0}^d b_i\mathcal{E}_iz^i
   =\sum_{i=0}^d c_i\mathcal{F}_iz^i.$$
$f$ is said random hyperbolic (respectively elliptic, flat)
when the $a_i$'s (respectively the $b_i$'s, $c_i$'s)
are random numbers.
%
\newcommand{\pev}{z}
\newcommand{\pevval}{\dot f}
\newcommand{\pevset}{Z}

\subsection{Implementation}

Our prototype implementation\footnote{available at \url{\pwpolyurl}}
in \clang
of the algorithms 
presented in this article
is
based on the \clang libraries \arb\footnote{\url{\arburl}} 
(which provides multiprecision 
ball arithmetic)
and \mpsolve\footnote{\url{\mpsolveurl}}
(see \cite{bini2014solving}, which provides 
a solver for univariate polynomials
relying on Ehrlich's, aka Aberth's, iterations
with 
multiprecision floating point arithmetic).

\pweval implements the evaluation process
described in Sec.~\ref{section_evaluation}
and \pwroots the
root isolation algorithm \ref{alg:rootfinding};
they
take in input 
a polynomial 
an $f$ 
and an $m$
and compute a piecewise approximation of $f$
using ball arithmetic.
\pwroots 
outputs a set of $d'$ pairwise
disjoints complex discs. If $d'=d$, each root of $f$ is isolated in a disc of the output.
Step \ref{line:factor} of algorithm \ref{alg:rootfinding}
is achieved 
by applying the solver using secular equations of \mpsolve.
\pweval also takes in input a set $\pevset$ of points in $\mathbb C$ 
and output for each $\pev\in\pevset$
a disc containing $f(\pev)$
of radius
less than $2^{-m} d \hat f(|\pev|)$.
Points are evaluated one by one
(without multi-point evaluation).




We introduced in our implementation a slight variation
of Algo.~\ref{alg:rings}: 
in step \ref{line:s},
we compute 
$s_{n+1}$ as 
$s_{n} + c\frac{m}{u_n-\ell_n+1}$
where $c$ is a real positive parameter in $O(1)$.
While preserving the bit complexity of the overall approach
this changes the subdivision of the complex plane 
in sectors and the piecewise approximation
in the following way: 
the greatest is $c$,
the widest are the rings,
the less are the numbers of angular sectors in rings
and the degrees of approximations 
and the fastest is 
the computation of the piecewise approximation
because
there are less approximations to compute.


Figures \ref{figure:newton_polygons} 
and
\ref{figure:sectors_partitions}
show the Newton polygons
and the sectors partitions of the complex plane
for hyperbolic, elliptic and flat random polynomials
obtained with our implementation with $c=1$.

We observed that the bottleneck in \pweval
was the computation of the piecewise approximation;
we choose $c=7/2$ in \pweval 
which 
makes this step faster.
In \pwroots, the dominant step
is the approximation of the roots 
of the approximating polynomials
with \mpsolve which complexity is quadratic in their degrees; we choose $c=2/5$
to produce 
approximations with smaller degrees.

\vspace*{-0.2cm}
\subsection{Numerical results}
\label{subsec_numerical_results}

All the times given below are sequential times in seconds on a \machine machine on Linux.

The input polynomials we consider in our tests are
$53$-bits floating point approximations
of random dense hyperbolic, elliptic and flat polynomials
which integers coefficient factors with a uniform law
in the interval $[2^{-8}, 2^8]$.
For root isolation,
we also consider $53$-bits floating point approximations of
resultant polynomials of two bivariate random dense hyperbolic polynomials.

We generate sets of random complex points
with rational real and imaginary parts
with numerator and denominators uniformly
chosen in $[2^{-16}, 2^{16}]$.

\vspace*{-0.15cm}
\subsubsection{Evaluation}


We used \pweval to evaluate 
random dense
hyperbolic, elliptic and flat polynomials
of increasing degree $d$ at $d$ points.
We compare the running time of \pweval
with the time required to evaluate 
$f$ at the same set of $d$ points with 
the 
a rectangular splitting algorithm 
(see \cite{brent2010modern, johansson2014evaluating})
available 
in \arb.
We choose $m$ (the output precision for 
\pweval) and the precision for
\arb so that the medians of the $\log_2$
of the errors relative to $\widetilde{f}$
of the evaluations 
is about $-50 \pm 10\%$.
In our tests, evaluations with Horner's rule 
were always slower than evaluations
with rectangular splitting.

Figure~\ref{figure:bench_eval} shows those
running times in an histogram
while detailing for \pweval
the time for computing the piecewise approximation and the time for evaluating it 
at the $d$ points (respectively called ``pw comp'' and ``pw eval'' in fig.~\ref{figure:bench_eval}).
\pweval becomes faster than the evaluation with rectangular splitting of $f$ for degrees above $40000$.
Notice also that once the piecewise approximation is computed, evaluating it at a single point is several orders of magnitude faster than one evaluation (at a single point) with rectangular splitting.


\vspace*{-0.15cm}
\subsubsection{Root isolation}

We used \pwroots to isolate the roots of 
the polynomials of our test suite
with $m=2(30+\lceil \log_2(d+1) \rceil)$ which
allowed \pwroots to always isolate all the roots.

We first compare the running times of \pwroots
and \hcroots
which is a \clang implementation of the root finding algorithm of \cite{moroz2022new}[Sec 5, Algo. 3].
\hcroots 
computes a piecewise approximation at a given precision $m$
which is dedicated to hyperbolic polynomials;
the root isolation process 
is embedded in a loop where $m$ is doubled while
all the roots are not isolated.
When $m$ reaches $d$, a subdivision 
root isolator with a complexity quadratic in $d$ is used.
We give in table \ref{table_pwpoly_hcroots} running times of \pwroots
and \hcroots 
for small values of $d$.
 For hyperbolic polynomials, the running times
 of both \hcroots and \pwroots
 are more or less linear in $d$;
 \hcroots is faster by a constant factor
 due to specificities of the implementations.
 \hcroots is not well adapted
 to other weights than the hyperbolic ones
 and in these cases its running time
 is dominated by the one of the quadratic solver for $f$, as depicted in table \ref{table_pwpoly_hcroots} (elliptic and flat cases).

%
%


 
 We also compare \pwroots and \mpsolve for hyperbolic, elliptic and flat random dense polynomials
 of degrees up to $20000$
 in
fig.~\ref{figure:bench_isolate}
where one can observe the linear
complexity in $d$ of \pwroots
(for the well conditionned polynomials
in consideration)
in contrast to the one of
\mpsolve which is quadratic in $d$.

Figure \ref{figure:bench_isolate_res} 
compares \pwroots
and \mpsolve
for isolating the roots of 
a resultant polynomial of degree $d$
of two bivariate random dense hyperbolic
polynomials of degree $\sqrt{d}$,
for $\sqrt{d}$ up to $100$.
\pwroots becomes faster than \mpsolve
for $d$ above $60^2$.

\bibliographystyle{plain}
\bibliography{references}

\appendix
\section{Benchmarks}

\begin{figure}[h]
  \includegraphics[width=\linewidth]{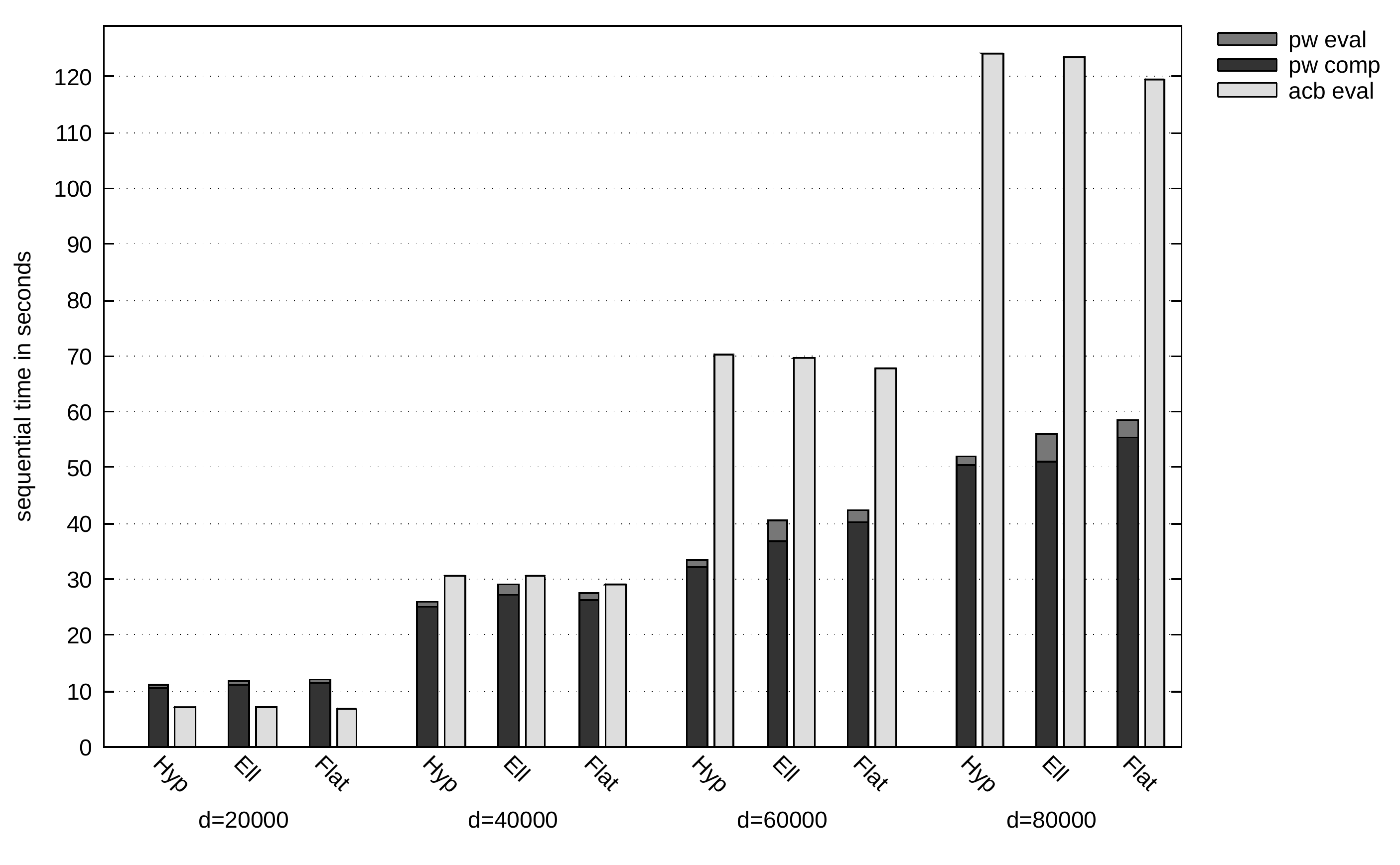}
  \vspace*{-0.5cm}
  \caption{ Evaluating random dense polynomials of degree
  $d$ at $d$ random points with \pweval and the rectangular splitting evaluation in \arb.
  }
  \label{figure:bench_eval}
\end{figure}

\begin{figure}[h]
  \includegraphics[width=\linewidth]{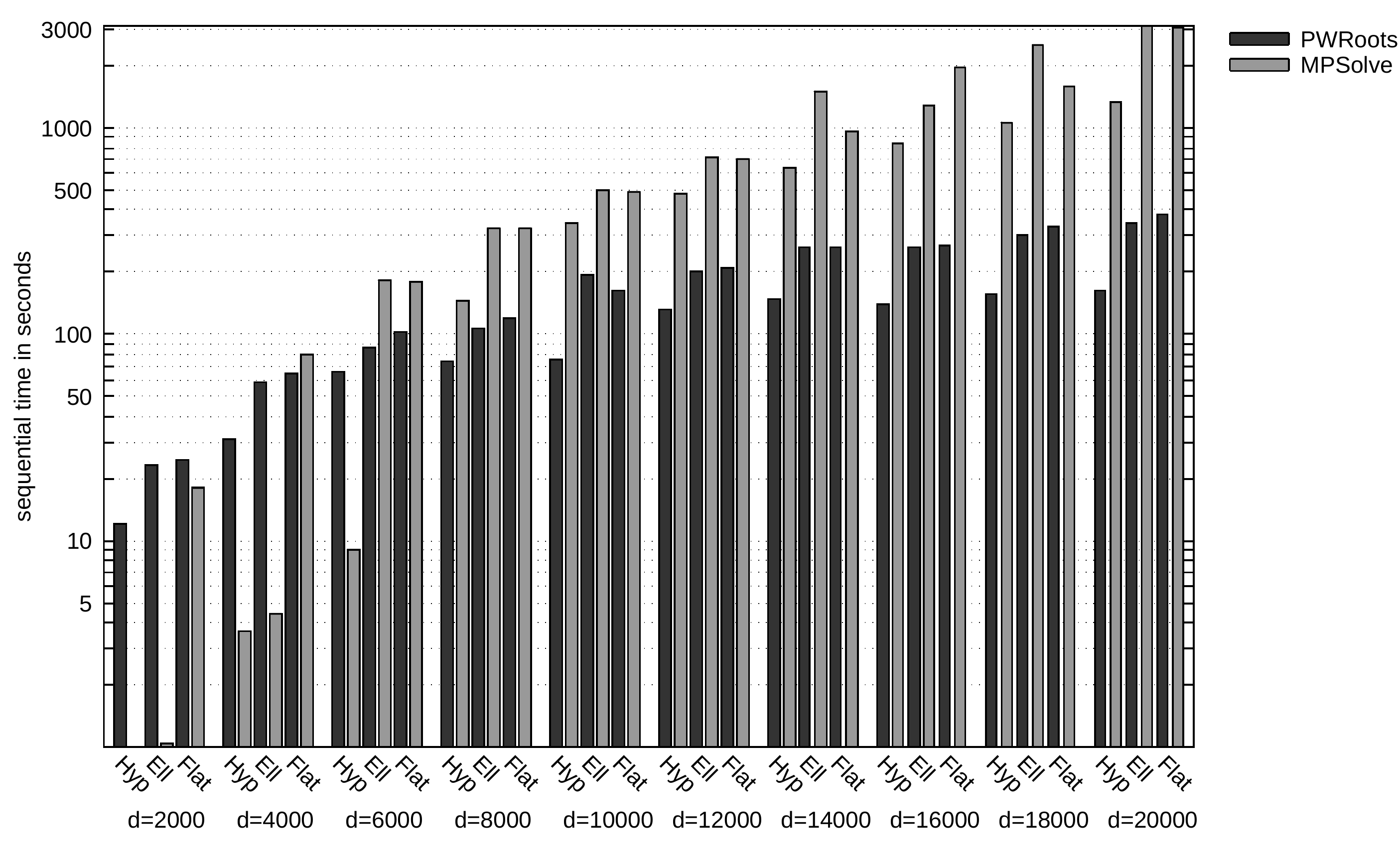}
  \vspace*{-0.5cm}
  \caption{Complex roots isolation of random dense 
  hyperbolic, elliptic and flat
  polynomials of increasing degrees $d$ with
  \pwroots and \mpsolve.}
  \label{figure:bench_isolate}
\end{figure}

\begin{table}[ht!]
\centering
\begin{tiny}
 \begin{tabular}{r||cc||cc||cc||cc||}
 & \multicolumn{2}{c||}{$d=200$}
 & \multicolumn{2}{c||}{$d=400$}
 & \multicolumn{2}{c||}{$d=800$}
 & \multicolumn{2}{c||}{$d=1600$} \\\hline
 & \hspace*{-0.1cm}\pwroots\hspace*{-0.1cm} & \hspace*{-0.1cm}\hcroots\hspace*{-0.1cm}
 & \hspace*{-0.1cm}\pwroots\hspace*{-0.1cm} & \hspace*{-0.1cm}\hcroots\hspace*{-0.1cm}
 & \hspace*{-0.1cm}\pwroots\hspace*{-0.1cm} & \hspace*{-0.1cm}\hcroots\hspace*{-0.1cm}
 & \hspace*{-0.1cm}\pwroots\hspace*{-0.1cm} & \hspace*{-0.1cm}\hcroots\hspace*{-0.1cm} \\\hline
 
 hyperbolic & 1.10 & 0.23 & 2.50 & 0.43 & 6.84 & 0.92 & 13.20 & 1.95 \\\hline
 elliptic   & 1.63 & 10.8 & 3.48 & 70.2 & 11.5 & 611 & 17.8 & $>999$ \\\hline
 flat       & 1.69 & 19.6 & 3.46 & 68.7 & 9.46 & 511  & 22.0 &  $>999$ \\\hline
 
 \end{tabular}
 \caption{
 Running times in seconds for isolating the complex roots
 of random dense 
  hyperbolic, elliptic and flat
  polynomials of increasing degrees $d$ with
  \pwroots and \hcroots.}
\label{table_pwpoly_hcroots}
\end{tiny}
\vspace*{-1cm}
\end{table}

\begin{figure}
  \includegraphics[width=\linewidth]{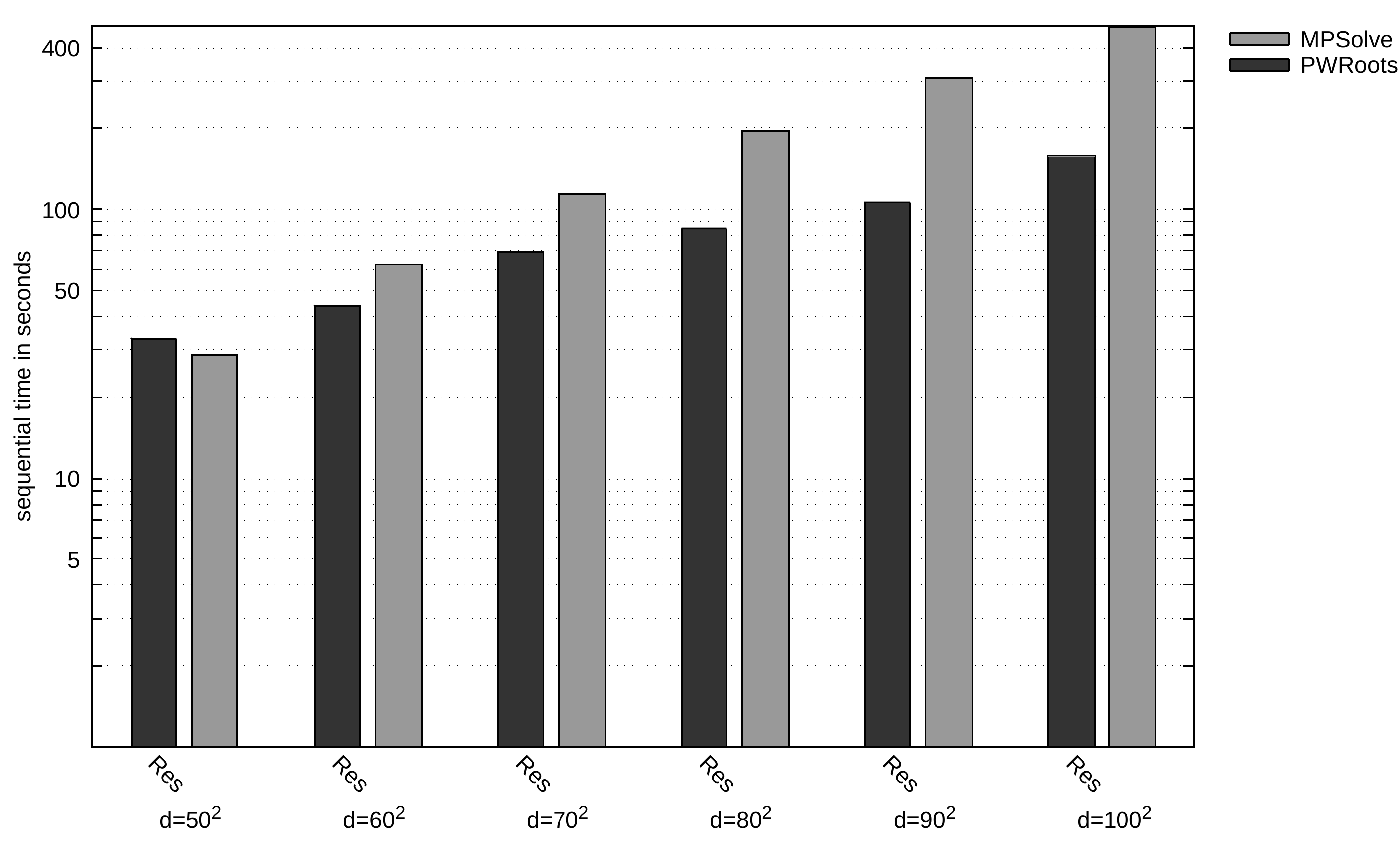}
  \vspace*{-0.5cm}
  \caption{Complex roots isolation of the
  resultant polynomial of two random dense 
  hyperbolic polynomials of increasing degrees $\sqrt{d}$ with
  \pwroots and \mpsolve.}
  \label{figure:bench_isolate_res}
\end{figure}


\begin{figure*}
  \parbox[b]{0.58\linewidth}{
    \hfill\includegraphics[width=0.964\linewidth]{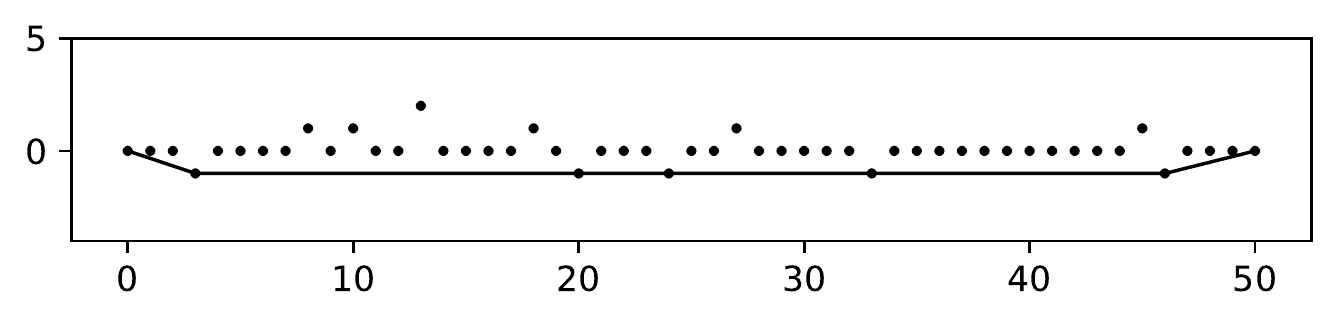}
    \includegraphics[width=\linewidth]{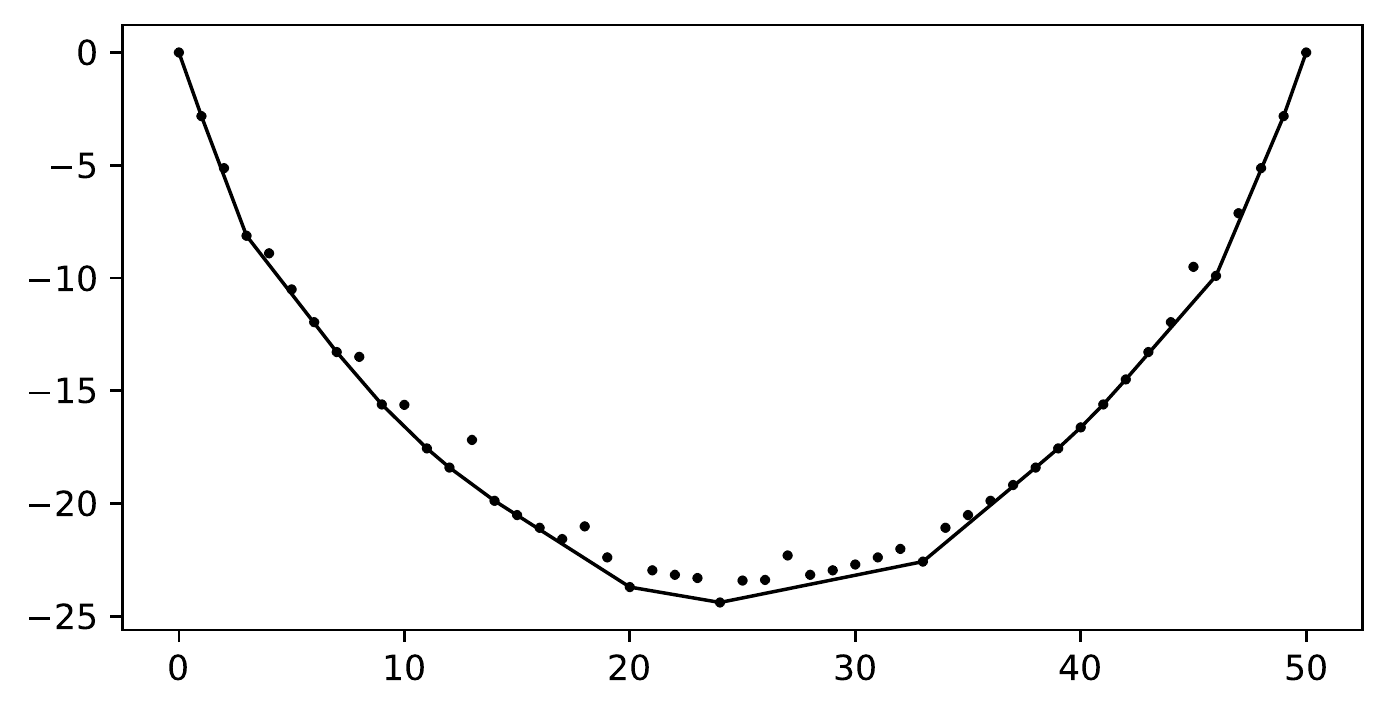}
  }
  \quad
  \includegraphics[height=24em]{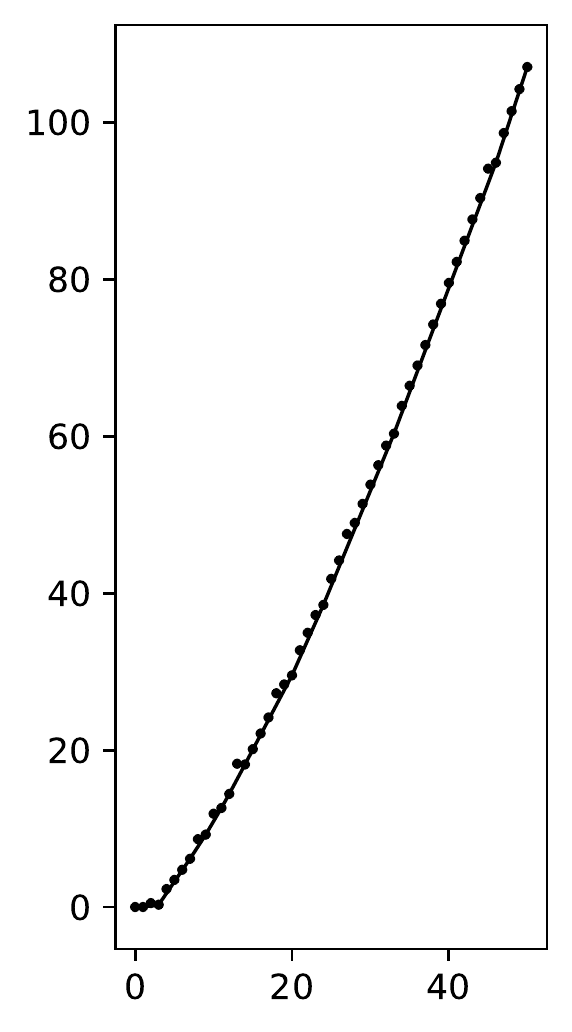}
\caption{Newton polygon of a hyperbolic (top left), elliptic (bottom
left), and flat (right)  polynomial of degree $50$} 
\label{figure:newton_polygons}
\end{figure*}

%
%


\begin{figure*}
  \includegraphics[width=0.3\linewidth]{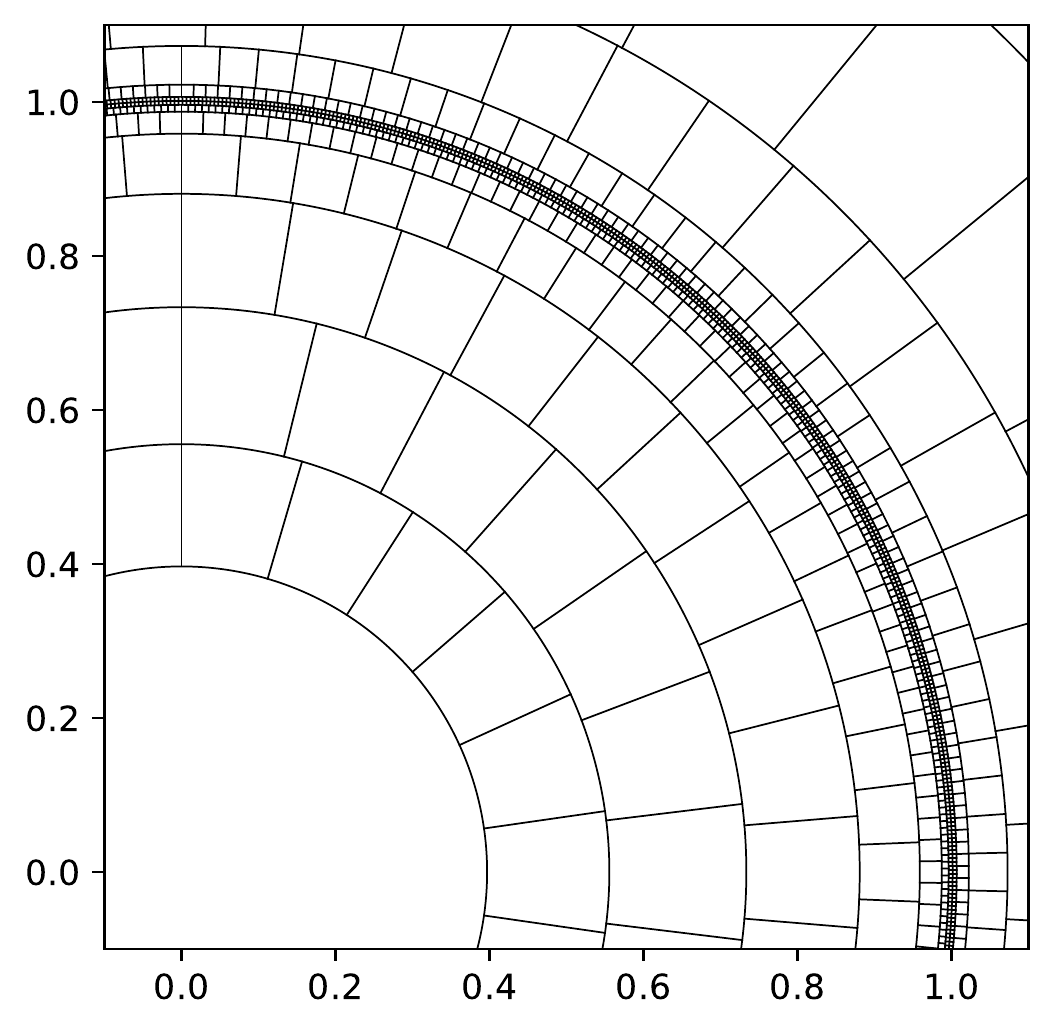}
  \quad
  \includegraphics[width=0.3\linewidth]{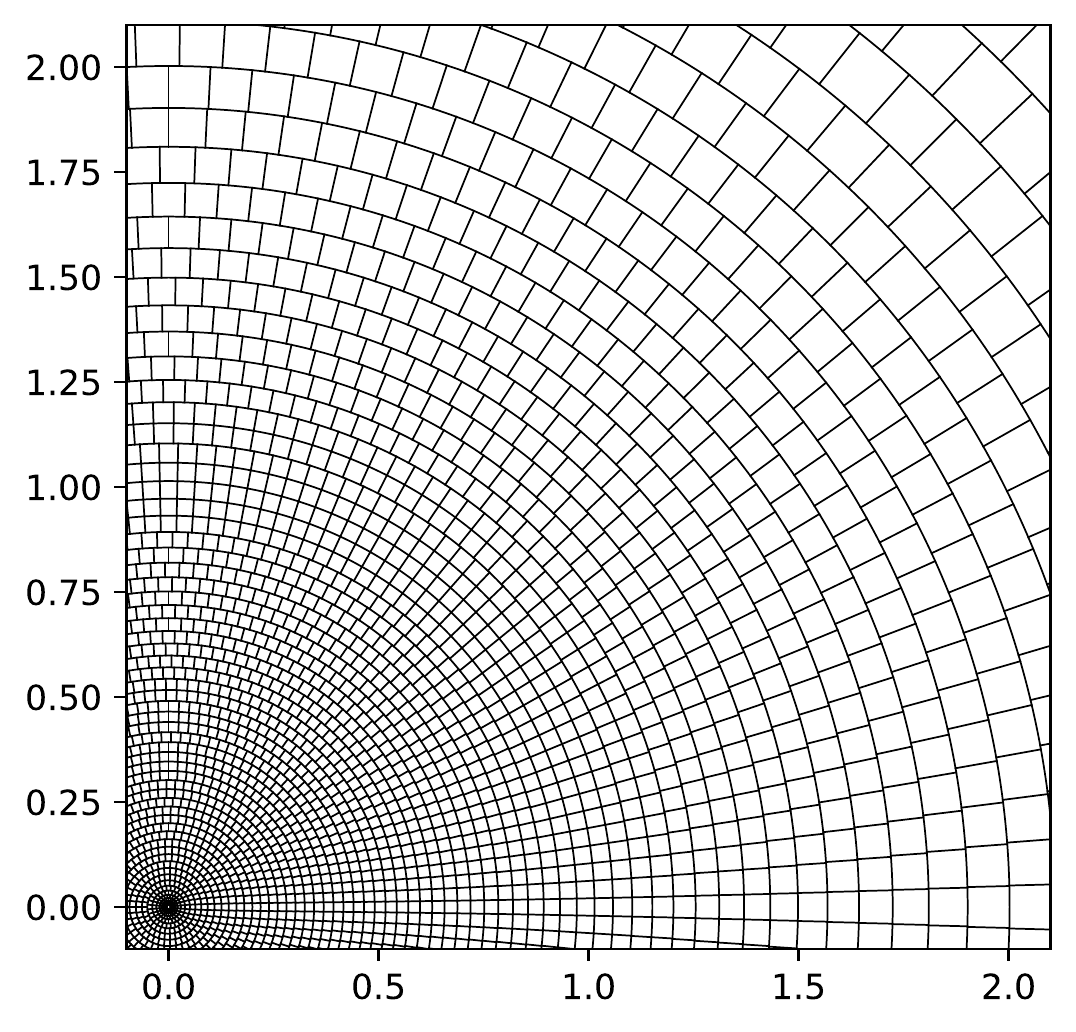}
  \quad
  \includegraphics[width=0.3\linewidth]{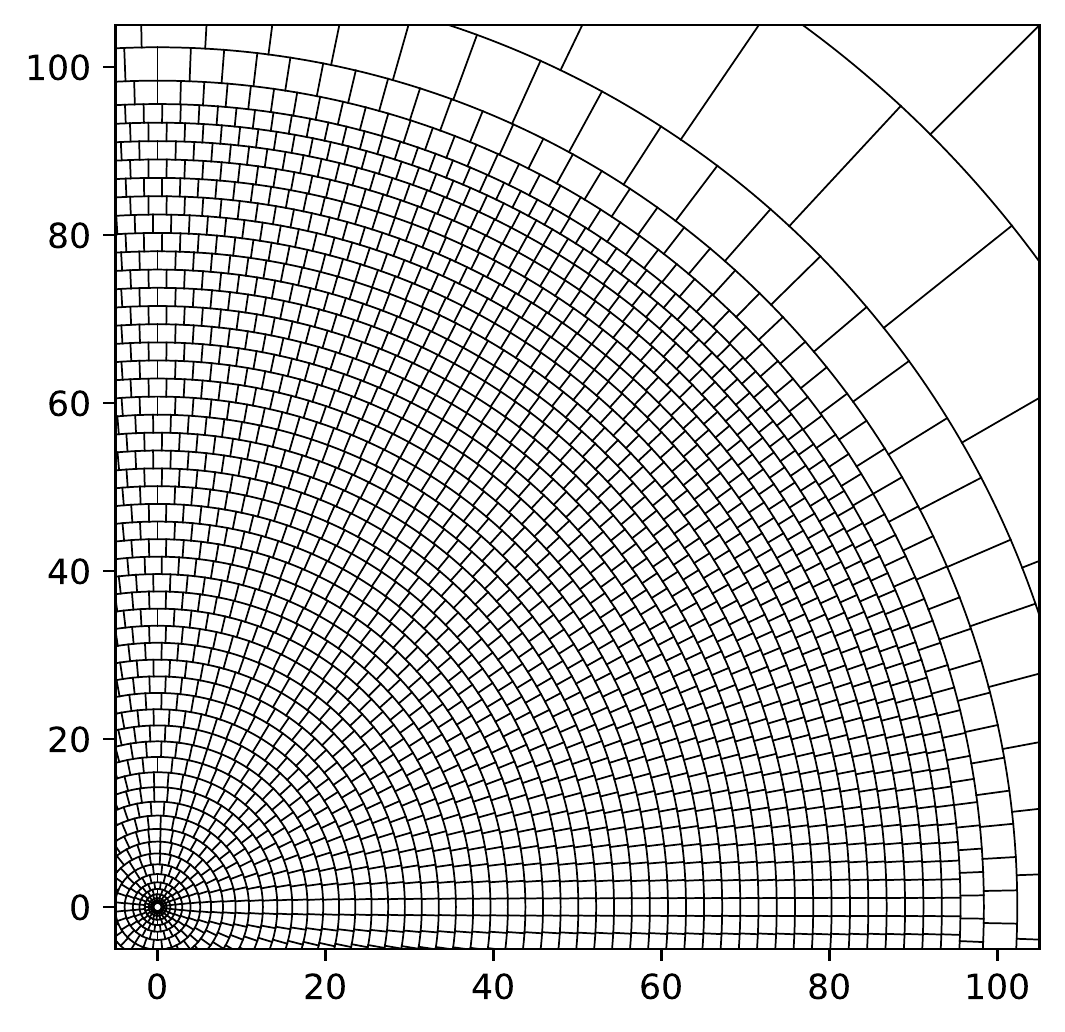}
  \caption{Sector partitions associated to a hyperbolic (left), elliptic (middle) and flat (right) polynomial of degree $10\,000$} 
  \label{figure:sectors_partitions}
\end{figure*}

\end{document}